%% file: main.tex
\documentclass[11pt,oneside,leqno,cmex10]{article}

\usepackage{amsfonts}
\usepackage{amssymb}
\usepackage{amsmath}

\usepackage{amsthm}

\usepackage{natbib}
\usepackage{enumerate}

\usepackage{graphicx}

\usepackage{bm}

\usepackage{wrapfig}
\usepackage{floatpag}

\usepackage{nicefrac}

\usepackage[colorlinks,citecolor=blue,urlcolor=blue]{hyperref}

\makeatletter
\renewcommand{\paragraph}{%
  \@startsection{paragraph}{4}%
  {\z@}{1ex \@plus 1ex \@minus .2ex}{-1em}%
  {\normalfont\normalsize\bfseries}%
}
\makeatother

\theoremstyle{plain}
\newtheorem{thm}{\protect\theoremname}
\theoremstyle{plain}

\theoremstyle{definition}

\theoremstyle{plain}
\newtheorem{cor}[]{\protect\corollaryname}
\theoremstyle{plain}

\theoremstyle{plain}

\theoremstyle{plain}
\newtheorem{lem}[]{\protect\lemmaname}
\theoremstyle{definition}
\newtheorem{rem}[]{\protect\remarkname}

\providecommand{\algorithmname}{Algorithm}
\providecommand{\assumptionname}{Assumption}
\providecommand{\corollaryname}{Corollary}
\providecommand{\definitionname}{Definition}
\providecommand{\propositionname}{Proposition}
\providecommand{\theoremname}{Theorem}
\providecommand{\remarkname}{Remark}
\providecommand{\lemmaname}{Lemma}

\newcommand{\footremember}[2]{%
\footnote{#2}
\newcounter{#1}
\setcounter{#1}{\value{footnote}}%
}
\newcommand{\footrecall}[1]{%
\footnotemark[\value{#1}]%
}

\begin{document}
\global\long\def\CC{\mathbb{C}}
\global\long\def\SS{S^{1}}
\global\long\def\RR{\mathbb{R}}
\global\long\def\actson{\curvearrowright}
\global\long\def\ra{\rightarrow}
\global\long\def\z{\mathbf{z}}
\global\long\def\ZZ{\mathbb{Z}}
\global\long\def\NN{\mathbb{N}}
\global\long\def\sgn{\mathrm{sgn}\:}
\global\long\def\RRpos{\RR_{>0}}
\global\long\def\var{\mathrm{var}}
\global\long\def\circint{\int_{-\pi}^{\pi}}
\global\long\def\F{\mathcal{F}}
\global\long\def\pb#1{\langle#1\rangle}
\global\long\def\op{\mathrm{op}}
\global\long\def\Op{\mathrm{op}}
\global\long\def\supp{\mathrm{supp}}
\global\long\def\ceil#1{\lceil#1\rceil}
\global\long\def\TV{\mathrm{TV}}
\global\long\def\floor#1{\lfloor#1\rfloor}
\global\long\def\vt{\vartheta}
\global\long\def\vp{\varphi}
\global\long\def\class#1{[#1]}
\global\long\def\of{(\cdot)}
\global\long\def\one{\mathbb{1}}
\global\long\def\cov{\mathrm{cov}}
\global\long\def\CC{\mathbb{C}}
\global\long\def\SS{S^{1}}
\global\long\def\RR{\mathbb{R}}
\global\long\def\actson{\curvearrowright}
\global\long\def\ra{\rightarrow}
\global\long\def\z{\mathbf{z}}
\global\long\def\ZZ{\mathbb{Z}}
\global\long\def\h{\mu}
\global\long\def\convr{*_{\RR}}
\global\long\def\x{\mathbf{x}}
\global\long\def\ve{\epsilon}
\global\long\def\cv{\mathfrak{c}}
\global\long\def\wh#1{\hat{#1}}
\global\long\def\norm#1{\left|\left|#1\right|\right|}
\global\long\def\Tmean{T}
\global\long\def\Tslope{m}
\global\long\def\degC#1{#1^{\circ}\mathrm{C}}
\global\long\def\X{\mathbf{X}}
\global\long\def\bbeta{\boldsymbol{\beta}}
\global\long\def\b{\mathbf{b}}
\global\long\def\Y{\mathbf{Y}}
\global\long\def\H{\mathbf{H}}
\global\long\def\e{\bm{\epsilon}}
\global\long\def\s{\mathbf{s}}
\global\long\def\t{\mathbf{t}}
\global\long\def\R{\mathbf{R}}
\global\long\def\dAc{\partial A_c}
\global\long\def\cl{\mathrm{cl}}
\global\long\def\Hoe{\text{H\"older}}
\global\long\def\bb#1{\mathbb{#1}}
\global\long\def\bX{\bm{X}}
\global\long\def\bY{\bm{Y}}
\global\long\def\D{\mathfrak{D}}
\global\long\def\dGH{d_{\mathcal{GH}}}
\global\long\def\dWp{d_{\mathcal{W},p}}
\global\long\def\umu{\underline{\mu}}
\global\long\def\FLB{\mathbf{FLB}}
\global\long\def\dX{\mathbf{d_\bX}}
\global\long\def\dY{\mathbf{d_\bY}}
\global\long\def\mX{{\mu_\bX}}
\global\long\def\mY{{\mu_\bY}}	
\global\long\def\cal#1{\mathbf{\mathcal{#1}}}
\global\long\def\UXn{U_{\bX n}}
\global\long\def\dXX{{\Delta_{\bX}}}
\global\long\def\dYY{{\Delta_{\bY}}}
\global\long\def\d{\:\mathrm{d}}
\global\long\def\Ds{\mathfrak{d}}
\global\long\def\ph{\hat{p}}
\global\long\def\muh{\hat{\mu}}
\global\long\def\onevec{\one}
\global\long\def\T{\mathcal{T}}
\global\long\def\root{\mathrm{root}}
\global\long\def\parent{\mathrm{parent}}
\global\long\def\children{\mathrm{children}}
\global\long\def\scp#1#2{\langle #1, #2 \rangle}
\global\long\def\X{\mathcal{X}}
\global\long\def\br{\bm{r}}
\global\long\def\bs{\bm{s}}
\global\long\def\brh{\hat{\bm{r}}}
\global\long\def\bsh{\hat{\bm{s}}}
\global\long\def\rh{\hat{r}}
\global\long\def\sh{\hat{s}}
\global\long\def\bu{{\bm{u}}}
\global\long\def\bv{\bm{v}}
\global\long\def\bh{\bm{h}}
\global\long\def\bw{\bm{w}}
\global\long\def\bD{{\bm{D}}}
\global\long\def\bx{\bm{x}}
\global\long\def\by{\bm{y}}
\global\long\def\bZ{\bm{Z}}
\global\long\def\WhnmB{\hat{W}^{(S)}}
\global\long\def\smax{{\sigma_{\mathrm{max}}}}
\global\long\def\sT{\sigma_\T}
\global\long\def\bsT{\bm{\sigma}_\T}
\global\long\def\E{\mathcal{E}}
\global\long\def\V{\mathcal{V}}
\global\long\def\cmax{c_{\mathrm{max}}}
\global\long\def\h{h}
\global\long\def\Id{\mathrm{Id}}
\global\long\def\N{\mathcal{N}}
\global\long\def\diam{\:\mathrm{diam}}
\global\long\def\ns{{n^{*}}}
\global\long\def\ps{{p^{*}}}
\global\long\def\Wt{W^{(t)}}
\global\long\def\BL1{\mathrm{BL}_1}
\global\long\def\I{\mathcal{I}}
\global\long\def\simD{\stackrel{D}{\sim}}
\global\long\def\d{\mathrm{d}}
\global\long\def\dpp{d^p\!}

\title{Inference for empirical Wasserstein distances on finite spaces}

\author{Max Sommerfeld\footremember{fbms}{\scriptsize Felix Bernstein Institute
  for
  Mathematical Statistics in the Biosciences and Institute for Mathematical
  Stochastics, University of G\"ottingen,
  Goldschmidtstra{\ss}e 7, 37077 G\"ottingen}
  \and 
  Axel Munk\footrecall{fbms} \footnote{\scriptsize Max Planck Institute for
    Biophysical
      Chemistry, Am Fa{\ss}berg 11, 37077 G\"ottingen}
       }

       \date{}

\maketitle 
 
\input{source/abstract}



\normalsize

\input{source/intro}

\input{source/derivatives}

\input{source/appl}

\input{source/discussion}

\scriptsize
\bibliographystyle{apalike}
\bibliography{finitewasser,R}
\normalsize

\input{source/appendix.tex}

\end{document}

%% file: source/abstract.tex
\begin{abstract}
  The Wasserstein distance is an attractive
  tool for data analysis 
  but statistical 
  inference is hindered by the lack of distributional limits.  To overcome this
  obstacle, for  probability measures
  supported on finitely many points, we derive the asymptotic distribution of
  empirical  Wasserstein distances as the optimal value of a
  linear program with random objective function.  This facilitates statistical
  inference (e.g. confidence intervals for sample based Wasserstein distances)
  in large
  generality.  Our proof is based on directional Hadamard
  differentiability.  Failure of the classical bootstrap and alternatives are
  discussed.
  The utility of the distributional results is illustrated on two data sets.
\end{abstract}

%% file: source/intro.tex
\section{Introduction}
\label{sec:intro}
 The
 \textit{Wasserstein distance} \citep{vasershtein_markov_1969},
also known as Mallows distance \citep{mallows_note_1972},
Monge-Kantorovich-Rubinstein distance in the physical
sciences \citep{kantorovich_space_1958,rachev_monge-kantorovich_1985, jordan_variational_1998},
earth-mover's distance in computer science \citep{rubner_earth_2000} or optimal
transport distance in optimization \citep{ambrosio_lecture_2003},
is one  of the most fundamental metrics on the space of probability
measures. Besides its prominence in probability
(e.g. \cite{dobrushin_prescribing_1970, gray_probability_1988}) and finance
(e.g. \cite{rachev_mass_1998}) it has deep connections to the asymptotic theory of
PDEs of diffusion type (\cite{otto_geometry_2001}, \cite{villani_topics_2003,
villani_optimal_2008} and references therein).
In a statistical setting it has mainly been used as a tool to
prove weak convergence in the context of limit laws (e.g.
\cite{bickel_asymptotic_1981, shorack_empirical_1986, Johnson2005,
dumbgen_approximation_2011, Dorea2012}) as
it metrizes weak
convergence together with convergence of moments. However, recently the
empirical (i.e. estimated from data) Wasserstein distance  has
also been recognized as a central quantity itself in many
applications, among them clinical trials \citep{Munk1998,
freitag_nonparametric_2007}, metagenomics \citep{evans_phylogenetic_2012},
medical imaging
\citep{Ruttenberg2013}, goodness-of-fit testing
\citep{Freitag2005, Barrio1999}, biomedical engineering
\citep{oudre_classification_2012},
computer vision  \citep{gangbo_shape_2000, ni_local_2009}, cell biology
\citep{orlova_earth_2016} and model validation
\citep{halder_model_2011}. The barycenter with respect to the Wasserstein metric
\citep{agueh_barycenters_2011} has
been shown to elicit important structure from complex data and to be a
promising tool, for example in deformable models
 \citep{boissard_distributions_2015, agullo-antolin_parametric_2015}.
It has also been used in  large-scale Bayesian inference to combine
posterior distributions from subsets of the data
\citep{srivastava_scalable_2015}.

Generally speaking three characteristics of the Wasserstein distance make it
particularly attractive for various applications. First, it incorporates a ground
distance on the space in question. This often makes it more adequate than
competing metrics such as total-variation or $\chi^2$-metrics which are
oblivious to any metric or similarity structure on the ground space. As an
example, the success of the Wasserstein distance in metagenomics applications
can largely be attributed to this fact (see \cite{evans_phylogenetic_2012} and
also our application in Section \ref{sub:appl_alt}).

Second, it has a clear and intuitive interpretation
as the amount of 'work' required to transform one probability distribution into
another and the resulting transport can be visualized (see Section
\ref{sub:FP}). This is also interesting in applications where probability
distributions are used to represent actual physical mass and spatio-temporal
changes have to be tracked. 

Third, it is well-established \citep{rubner_earth_2000} that the Wasserstein
distance performs exceptionally well at capturing human perception of
similarity. This motivates its popularity in computer vision and related fields.

Despite these advantages, the use of the empirical Wasserstein distance in a
statistically
rigorous way is severely hampered by a lack of
inferential tools. We argue that this issue stems from considering
too large classes of candidate distributions (e.g. those which are absolutely
continuous with respect to the Lebesgue
measure if the ground space has dimension $\geq 2$).
In this paper, we therefore discuss the 
Wasserstein distance on finite spaces, which allows to solve this issue.
We argue that the restriction to finite spaces is not merely an approximation to
the truth, but rather that
this setting is sufficient for many practical
situations as measures often already come naturally discretized (e.g. two- or
three-dimensional images - see also our applications in Section \ref{sec:appl}).

We remark that from our methodology further inferential procedures can be
derived, e.g. a (M)ANOVA type of analysis and multiple comparisons of
Wasserstein distances based on their $p$-values (see e.g.
\cite{benjamini_controlling_1995}). Our techniques also extend immediately to
dependent samples $(X_i,Y_i)$ with marginals $\br$ and $\bs$.

\paragraph{Wasserstein distance}
Let $(\X,d)$ be a complete metric space with metric $d:\X\times\X\ra\RR_{\geq 0}$.  The
\textit{Wasserstein distance of order $p$} ($p\geq 1$) between 
two Borel probability measures $\mu_1$ and $\mu_2$ on $\X$ is defined as
\[
  W_p(\mu_1, \mu_2) = \left\{ \inf_{\nu\in\Pi(\mu_1, \mu_2)} \int_{\X\times \X} \dpp(x,x') \nu(dx, dx')
\right\}^{1/p},
\]
where $\Pi(\mu_1,\mu_2)$ is the set of  all Borel probability measures  on
$\X\times \X$ with marginals $\mu_1$ and $\mu_2$, respectively.

\paragraph{Wasserstein distance on finite spaces}
If we restrict in the above definition $\X = \left\{ x_1,\dots,x_N \right\}$ to
be a finite space, every
probability measure on $\X$ is given by a vector $\br$ in  
  $\mathcal{P}_\X = \left\{ \br = (r_x)_{x\in\X} \in\RR_{> 0}^\X : 
  \sum_{x\in\X} r_x =1 \right\}$,
via $P_{\br}(\{x\}) = r_x$. We will not distinguish between the vector
$\br$ and the measure it defines.
The \textit{Wasserstein distance of order $p$} between two finitely supported probability measures
$\br,\bs \in\mathcal{P}_\X$  then becomes
\begin{equation}
  W_p(\br, \bs) = 
  \left\{  
    \min_{\bw\in\Pi(\br, \bs)}  \sum_{x,x'\in\X} \dpp(x, x') w_{x,x'}
\right\}^{1/p},
  \label{eq:def_wass}
\end{equation}
where $\Pi(\br, \bs)$ is the set of all probability measures  on $\X\times\X$
with marginal distributions $\br$ and $\bs$, respectively. 
All our methods and  results concern this Wasserstein distance on finite spaces.

\subsection{Overview of main results}
\paragraph{Distributional limits}
The basis for inferential procedures for the Wasserstein distance on finite
spaces is a limit theorem for its empirical version $W_p(\brh_n, \bsh_m)$.
Here, the empirical measure generated by independent
random variables $X_1,\dots, X_n\sim\br$  is given by
$ \brh_n  = (\rh_{n,x})_{x\in\X}$,   where $\rh_{n,x}= \frac{1}{n} \#\left\{ k :
X_k=x \right\}$.
Let $\bsh_m$ be generated from i.i.d. $Y_1,\dots,Y_m\sim \bs$ in the same
fashion.
Under the null hypothesis $\br = \bs$ we prove that
\begin{equation}
  \left(\frac{nm}{n+m}\right)^{\frac{1}{2p}}  W_p(\brh_n,\bsh_m)  \Rightarrow
        \left\{ \max_{\bu \in \Phi^*_p}   \scp{\bm{G}}{\bu}
      \right\}^{\frac{1}{p}}, \quad n,m\ra \infty.
  \label{eq:first_mention_distr_lim}
\end{equation}
Here, '$\Rightarrow$' means convergence in distribution, $\bm{G}$ is a mean zero
Gaussian random vector with covariance depending
on $\br=\bs$ and $\Phi_p^*$ is the convex set of dual solutions to the
Wasserstein problem  depending on the metric $d$ only (see Theorem
\ref{thm:full}). 
In Section \ref{sub:FP} we use this result to assess the
statistical significance of the differences between real and synthetically
generated fingerprints in the Fingerprint
Verification Competition \citep{maio_fvc2002}.

We give analogous results  under the alternative
$\br\neq\bs$. This extends the scope of our results beyond the classical
two-sample (or goodness-of-fit test) as it allows for confidence statements on
$W_p(\br, \bs)$ when the null hypothesis of equality is likely or even
\textit{known to be false}.
An example for this is given by our application to metagenomics (Section
\ref{sub:appl_alt}) where samples from the same person taken at different times
are typically statistically different but our asymptotic results allow us to
assert with statistical significance that inter-personal distances are larger
that intra-personal ones.

\paragraph{Proof strategy} We prove these results by showing that the
Wasserstein distance is
\textit{directionally Hadamard differentiable} \citep{Shapiro1990} and the right
hand side of \eqref{eq:first_mention_distr_lim} is its derivative evaluated at
the Gaussian limit of the empirical multinomial process (see Theorem
\ref{thm:derivative_Wasserstein}). This notion generalizes Hadamard
differentiability by allowing \textit{non-linear} derivatives but still allows
for a refined delta-method (\cite{Romisch2004} and Theorem
\ref{thm:delta_method}).
Notably, the Wasserstein distance is not Hadamard differentiable in the usual
sense.

\paragraph{Explicit limiting distribution for tree metrics}
When the space $\X$ are the vertices of a tree and the metric $d$ is
given by path length we give an explicit expression for the limiting
distribution in \eqref{eq:first_mention_distr_lim} (see Theorem
\ref{thm:trees}).
In contrast to the general case, this explicit formula allows for fast
and direct simulation of the  limiting distribution.
This extends a previous result of \cite{Samworth2004} who considered a finite
number of point masses on the real line.
The Wasserstein distance on trees has, to the best of our knowledge, only been
considered in two papers:
\cite{kloeckner_geometric_2013} studies the geometric properties of the
Wasserstein space of measures on a tree  and 
\cite{evans_phylogenetic_2012} use the Wasserstein distance on phylogenetic
trees to compare microbial communities. 
\paragraph{The bootstrap}
Directional Hadamard
differentiability is not enough to guarantee the consistency of the naive
($n$ out of $n$) bootstrap
\citep{dumbgen_nondifferentiable_1993, fang_inference_2014}  -
in contrast to the usual notion of Hadamard differentiability. This
implies that the bootstrap is \textit{not} consistent for the Wasserstein
distance \eqref{eq:def_wass}(see Theorem \ref{thm:boot}).
In contrast, the $m$-out-of-$n$ bootstrap for $m/n\ra 0$ is known to be
consistent
in this setting \citep{dumbgen_nondifferentiable_1993} and can be applied to the
Wasserstein distance.
Under the null hypothesis $\br=\bs$, however, there is a more direct way of
obtaining an approximation of the limiting distribution.
In the appendix, we discuss this alternative re-sampling scheme based on ideas
of
\cite{fang_inference_2014}, that essentially consists of plugging
in a bootstrap version of the underlying empirical process in the derivative. We
show that this scheme, which we will call \textit{directional bootstrap}, is
consistent for the Wasserstein distance (see Theorem \ref{thm:boot}, Section
\ref{sub:boot}). 

\subsection{Related work}
\paragraph{Empirical Wasserstein distances}
In very general terms, we study a particular case (finite spaces) of the following question and its two-sample analog: Given the  empirical
measure
$\mu_n$  based on $n$  i.i.d. random variables taking variables in a metric
space
with law $\mu$. What can be inferred about
$W_p(\mu_n,\mu_0)$ for a reference measure $\mu_0$ which may be equal to
$\mu$? 

It is a well-known and straightforward consequence of the strong law of large
numbers  that if the $p$-th moments are finite for $\mu$ and $\mu_0$ 
 then $W_p(\mu_n, \mu_0)$ converges to $W_p(\mu, \mu_0)$,
almost surely, as the sample size $n$  approaches infinity
{\citep[Cor. 6.11]{villani_optimal_2008}}. 
Determining the
exact rate of this convergence is the subject of an impressive body of
literature developed over the last decades starting with the seminal work of
\cite{ajtai_optimal_1984} considering for $\mu_0$ the uniform distribution on the unit square,
followed by \cite{talagrand_matching_1992, talagrand_transportation_1994} for the uniform distribution in
higher dimensions and \cite{Horowitz1994} giving bounds on mean rates
of convergence. 
\cite{Boissard2014,fournier_rate_2014} gave
general deviation inequalities for the empirical Wasserstein distance on metric
spaces. 
For a discussion in the light of our distributional limit results see Section \ref{sec:dis}.

Distributional limits give a natural perspective for practicable inference,
but
despite considerable interest in the topic  have remained elusive to a large
extent. For measures on $\X = \RR$ a rather complete theory is available (see \cite{Munk1998,
freitag_nonparametric_2007, Freitag2005} for $\mu_0\neq \mu$
and e.g. \cite{Barrio1999,
samworth_empirical_2005, Barrio2005} for $\mu_0 = \mu$ as well as
\cite{mason_weighted_2016,bobkov_one-dimensional_2014} for recent surveys).
However, for $\X = \RR^d$,
$d\geq 2$ the only distributional result known to us is due to \cite{rippl_limit_2015}
for specific multivariate (elliptic) parametric classes of distributions, when
the empirical measure is replaced by a parametric estimate. 
In the context of deformable models distributional results are proven
\citep{del_barrio_statistical_2015} for specific  multidimensional
parametric models which factor into one-dimensional parts.

The simple reason why the Wasserstein distance is so much easier to handle in
the one-dimensional case is that in this case the optimal coupling attaining the
infimum in \eqref{eq:def_wass} is known explicitly. In fact, the  Wasserstein
distance of order $p$ between two measures on
$\RR$ then becomes the $L^p$ norm of the difference of their quantile functions
(see \cite{mallows_note_1972} for an early reference) and
the analysis of empirical Wasserstein distances can be based on quantile process
theory. 
Beyond this case, explicit coupling results are only known for multivariate
Gaussians
and elliptic distributions \citep{Gelbrich1990}.
A classical result of
\cite{ajtai_optimal_1984} for the uniform distribution on $\X = [0,1]^2$
suggests
that, even in this simple case, distributional limits will have a complicated
form if they exist at all. We will elaborate on this thought in the discussion, in
Section \ref{sec:dis}. 

The Wasserstein distance on finite spaces has been considered recently  by
\cite{gozlan_displacement_2013} to derive entropy inequalities on graphs and by
\cite{erbar_ricci_2012} to define Ricci curvature for Markov chains on discrete
spaces. To the best of our knowledge, empirical
Wasserstein distances on finite spaces have only been considered by
\cite{Samworth2004} in the special case of measures supported on $\RR$. We will
show (Section \ref{sec:trees}) that our results extend theirs.

\paragraph{Directional Hadamard differentiability}
We prove our distributional limit theorems using the theory of parametric
programming  \citep{bonnans_perturbation_2013} which
investigates how the optimal value and the optimal solutions of an
optimization problem change when the objective function and the constraints
are changed. 
While differentiability properties of optimal values of linear programs are
extremely well studied such results have, to the best of our knowledge, not yet
been applied to the statistical analysis of Wasserstein distances. 

It is well-known  that under certain
conditions the optimal value of a mathematical program is differentiable with
respect to the constraints
of the problem \citep{Rockafellar1984, Gal1997}. However, the  derivative will
typically be non-linear. The appropriate concept for this is directional Hadamard
differentiability \citep{Shapiro1990}.
The derivative of the optimal
value of a mathematical program is typically again given as an extremal value.

Although the delta-method
 for directional Hadamard
derivatives has been known for a long time \citep{shapiro_asymptotic_1991, dumbgen_nondifferentiable_1993}, this
notion scarcely appears in the statistical context (with some exceptions,
such as \cite{Romisch2004}, see also \cite{donoho_pathologies_1988}). Recently, an interest in the topic has evolved in
econometrics (see \cite{fang_inference_2014} and references therein). 

\paragraph{Organization of the paper}
In Section \ref{sec:asymp_distr} we give a comprehensive result on distributional
limits for the Wasserstein distance for measures supported on finitely many
points. In Section \ref{sec:appl} we apply our methods to two data sets to
highlight different aspects.
In Section \ref{sec:dis} we briefly address limitations and possible extensions
of our work.
In the supplementary Material we discuss the bootstrap for the Wasserstein
distance and give some technical proofs.

%% file: source/derivatives.tex
\section{Distributional limits}
\label{sec:asymp_distr}
\subsection{Main result}
In this section we give a comprehensive result on distributional limits for the
Wasserstein distance when the underlying population measures are supported on
finitely many points $\X = \left\{ x_1, \dots, x_N \right\}$. 
We denote the inner product on the vector space $\RR^\X$ by
 $\scp{\bu}{\bu'}=\sum_{x\in \X}u_x u'_{x}$ for $\bu,\bu' \in \RR^\X$. 
\begin{thm}
  \label{thm:full}
  Let $p\geq 1$,  $\br, \bs\in\mathcal{P}_\X$ and $\brh_n, \bsh_m$ generated by
  i.i.d.
  samples $X_1,\dots, X_n\sim\br$ and
  $Y_1,\dots, Y_m\sim\bs$, respectively. We
  define the convex sets
  \begin{equation}
    \begin{split}
      \Phi^*_p &= 
      \left\{ 
        \bu\in\RR^\X: 
        u_x - u_{x'} \leq d^p(x,x'), \quad x,x'\in\X
      \right\} \\
      \Phi^*_p(\br, \bs) &= 
      \left\{ 
        (\bu,\bv)\in\RR^\X \times \RR^\X : 
        \begin{split}
          &\scp{\bu}{\br} + \scp{\bv}{\bs}= W_p^p(\br, \bs), \\ 
          &u_x + v_{x'} \leq d^p(x, x'), \: x,x'\in\X        
        \end{split}
      \right\} 
    \end{split}
    \label{eq:sets}
  \end{equation}
and the multinomial covariance matrix
  \begin{equation}
    \Sigma(\br) = 
    \begin{bmatrix}
      r_{x_1} (1 - r_{x_1}) & -r_{x_1} r_{x_2} & \cdots & -r_{x_1} r_{x_N} \\
      -r_{x_2} r_{x_1} & r_{x_2} (1 - r_{x_2}) & \cdots & -r_{x_2} r_{x_N} \\
      \vdots & & \ddots &\vdots \\
      - r_{x_N} r_{x_1}& -r_{x_N} r_{x_2} & \cdots & r_{x_N} (1 - r_{x_N}) 
    \end{bmatrix}
    \label{eq:def_Sigma}
  \end{equation}
  such that with  independent Gaussian random variables
  $\bm{G}\sim\mathcal{N}(0,\Sigma(\br))$ and $\bm{H}\sim\mathcal{N}(0,
  \Sigma(\bs))$ we have the following.
  \begin{enumerate}[a)]
    \item \textbf{(One sample - Null hypothesis)}  With the sample size $n$
      approaching infinity, we have the weak convergence
      \begin{equation}
        n^{\frac{1}{2p}} W_p(\brh_n, \br) \Rightarrow 
        \left\{ \max_{\bu \in \Phi^*_p}   \scp{\bm{G}}{\bu}  \right\}^{\frac{1}{p}}.
        \label{eq:one_sample_null}
      \end{equation}
    \item \textbf{(One sample - Alternative)} With $n$ approaching infinity we have
      \begin{equation}
        n^{\frac{1}{2}} \left(W_p(\brh_n, \bs) - W_p(\br, \bs)\right) \Rightarrow 
        \frac{1}{p} W_p^{1-p}(\br, \bs)
        \left\{ \max_{(\bu,\bv) \in \Phi_p^*(\br, \bs)}   \scp{\bm{G}}{\bu}
      \right\}.
      \label{eq:one_sample_alt}
    \end{equation}
  \item \textbf{(Two samples - Null hypothesis)} Let $\rho_{n,m} = \left(
    nm/(n+m) \right)^{1/2}$. If $\br = \bs$ and $n$ and
    $m$ are approaching infinity such that $n\wedge m\ra\infty$ and
    $m/(n + m)\ra\lambda\in(0,1)$ we have 
    \begin{equation}
      \rho_{n,m}^{1/p}W_p(\brh_n, \bsh_m)
      \Rightarrow
        \left\{ \max_{\bu\in\Phi_p^*} \scp{\bm{G}}{\bu}
      \right\}^{\frac{1}{p}}.
      \label{eq:two_sample_null}
      \end{equation}
    \item \textbf{(Two samples - Alternative)}  With $n$ and $m$
      approaching infinity such that $n\wedge m\ra\infty$ and
      $m/(n + m)\ra\lambda\in[0,1]$ 
    \begin{equation}
      \begin{split}
        \rho_{n,m} & \left( 
        W_p(\brh_n, \bsh_m)  - W_p(\br,\bs) \right)
        \Rightarrow \\
        & \frac{1}{p} W_p^{1-p}(\br, \bs)
        \left\{ \max_{(\bu,\bv)\in\Phi_p^*(\br, \bs)} \sqrt{\lambda}\scp{\bm{G}}{\bu} +
        \sqrt{1 - \lambda}\scp{\bm{H}}{\bv} \right\}.
      \end{split}
        \label{eq:two_sample_alt}
      \end{equation}
  \end{enumerate}
\end{thm}
The sets $\Phi_p^*$ and $\Phi_p^*(\br, \bs)$ are (derived from) the dual
solutions to the Wasserstein linear program (see Theorem
\ref{thm:derivative_Wasserstein} below).
This result is valid for all probability measures with finite support,
regardless of the (dimension of the) underlying space. In particular, it
generalizes a result of \cite{Samworth2004}, who considered a finite collection
of point masses on the real line and $p=2$. We will re-obtain their result as a
special case in Section \ref{sec:trees} when we give explicit expressions for the limit
distribution when the metric $d$, which enters the limit law via the dual
solutions $\Phi_p^*$ or $\Phi_p^*(\br,\bs)$, is given by a tree. 

\begin{rem}
  In our numerical experiments (see Section \ref{sec:appl} we have found the
  representation \eqref{eq:two_sample_alt} to be numerically unstable when used
  to simulate from the limiting distribution under the alternative. We therefore
  give an alternative representation \eqref{eq:altalt} in the supplementary
  material as a one-dimensional optimization problem of a non-linear function
  (in contrast to a high-dimensional linear program shown here).
  Note that the limiting distribution under the null does not suffer from this
  problem and can be simulated from directly using a linear program solver.
\end{rem}

The scaling rate in Theorem \ref{thm:full} depends solely on $p$ and is
completely independent of the underlying space $\X$. This contrasts known bounds
on the rate of convergence in the continuous case. We will elaborate on the
differences in the discussion.
Typical choices are $p=1,2$. The faster scaling rate can be a reason to favor
$p=1$. In our numerical experiments however, this advantage was frequently
outweighed by larger quantiles of the limiting distribution. 

\cite{dumbgen_nondifferentiable_1993}
showed that the naive
$n$-out-of-$n$ bootstrap is inconsistent for functionals with a non-linear
Hadamard derivative, but 
resampling fewer than $n$ observations leads to a consistent bootstrap.
Since we will show in the following that the Wasserstein distance belongs to
this class of functionals, it is a direct consequence that the naive bootstrap
fails for the Wasserstein distance (see Section \ref{sub:boot} in the
supplementary material for details) and that the following holds.
\begin{thm}
  \label{thm:mofn}
    Let $\brh_n^*$ and $\bsh_m^*$ be bootstrap versions
      of $\brh_n$ and $\bsh_m$ that are obtained via re-sampling $k$ observations
      with $k/n\ra 0$ and $k/m\ra 0$. 
      Then, the plug-in bootstrap with
      $\brh_n^*$ and $\bsh_m^*$ is consistent, that is 
      \begin{equation*}
        \begin{split}
          \sup_{f\in\BL1(\RR)}E\left[ f(\phi_p(\sqrt{k}\left\{ (\brh_n^{**},
            \bsh_m^{**}) - (\brh_n, \bsh_m)
          \right\})) | X_1, \dots,X_n,Y_1, \dots, Y_m  \right] \\
          - E\left[ f\left(\rho_{n,m}\left\{ W_p^p(\brh_n, \bsh_m) - W_p^p(\br,
          \bs) \right\}\right) \right]
        \end{split}
      \end{equation*}
      converges to zero in probability.
\end{thm}
In the following we will prove our main Theorem \ref{thm:full} by 
\begin{enumerate}[i)]
  \item introducing  Hadamard directional differentiability, which
    does not require the derivative to be linear but still allows for a delta-method;
  \item  showing that  the map $(\br, \bs)\mapsto W_p(\br, \bs)$ is
    differentiable in this sense.
\end{enumerate}

\subsection{Hadamard directional derivatives} 
In this section we  follow \cite{Romisch2004}.  A map $f$ defined on a subset
$D_f\subset \RR^d$ with values in $\RR$ is called
\textit{Hadamard directionally differentiable} at $\bu\in\RR^d$ if there exists a
map $f'_{\bu} : \RR^d\ra \RR$ such that 
\begin{equation}
  \lim_{n\ra\infty} \frac{f(\bu+t_n\bh_n) - f(\bu)}{t_n} = f'_\bu(\bh)
  \label{eq:def_hadamard}
\end{equation}
for any $\bh\in\RR^d$ and for arbitrary sequences $t_n$ converging to zero from
above and $\bh_n$ converging to $\bh$ such that $\bu+t_n\bh_n\in D_f$ for all $n\in\NN$.
Note that in contrast to the usual notion of Hadamard differentiability (e.g.
\cite{van_der_vaart_weak_1996}) the derivative $\bh\mapsto f'_\bu(\bh)$ is \textit{not} required to be
linear. 
A prototypical example  is the absolute value $f:\RR\ra\RR$,
$t \mapsto |t|$ which is not in the usual sense Hadamard differentiable at
$t=0$ but directionally differentiable with the non-linear derivative $t\mapsto
|t|$.
\begin{thm}[{\citealp[Theorem 1]{Romisch2004}}]
  \label{thm:delta_method}
  Let $f$ be a function defined on a subset $F$ of
  $\RR^d$ with values in $\RR$, such that
  \begin{enumerate}
    \item $f$ is Hadamard directionally differentiable at $\bu\in F$  with
      derivative $f'_\bu:F\ra \RR$ and 
    \item there is a sequence of $\RR^d$-valued random variables $X_n$ and a
      sequence of non-negative numbers $\rho_n\ra\infty$ such that $\rho_n(X_n -
      \bu)\Rightarrow X$ for some random variable $X$ taking values in
      $F$.
  \end{enumerate}
  Then, $\rho_n(f(X_n) - f(\bu))\Rightarrow f_\bu'(X)$. 
\end{thm}

\subsection{Directional derivative of the Wasserstein distance}
 In this section we  show that the functional
$(\br, \bs)\mapsto W_p^p(\br, \bs)$ is Hadamard directionally 
differentiable and give a formula for
the derivative. 
 
The \textit{dual} program (cf.
{\cite[Ch. 4]{Luenberger2008}}, also \cite{kantorovich_space_1958}) of the linear program defining the Wasserstein
distance \eqref{eq:def_wass} is given by 
\begin{equation}
  \label{eq:dual}
  \begin{split}
    &\max_{(\bu, \bv)\in\RR^\X\times \RR^\X}\quad  \scp{\bu}{\br} + \scp{\bs}{\bv} \\
    \textbf{s.t.} \quad& u_x + v_{x'} \leq d^p(x, x') \quad \forall x,x'\in\X.
  \end{split}
\end{equation}
As noted above, the optimal value of the primal problem is  $W_p^p(\br, \bs)$
and by standard duality
theory of linear programs (e.g. \cite{Luenberger2008}) this is also the optimal
value of the dual problem. Therefore, the set of optimal solutions to the dual
problem is given by $\Phi^*_p(\br, \bs)$ as defined in \eqref{eq:sets}.
\begin{thm}
  \label{thm:derivative_Wasserstein}
  The functional $(\br, \bs)\mapsto W_p^p(\br, \bs)$ is directionally Hadamard
  differentiable at all $(\br, \bs)\in\mathcal{P}_\X\times\mathcal{P}_\X$
   with derivative
  \begin{equation}
    (\bh_1,\bh_2) \mapsto \max_{(\bu, \bv)\in\Phi^*_p(\br, \bs)}
    - (\scp{\bu}{\bh_1} + \scp{\bv}{\bh_2}).
    \label{eq:derivative}
  \end{equation}
\end{thm}
We can give a more explicit expression for the set $\Phi^*_p(\br, \bs)$ in the
case $\br = \bs$, when the optimal value of the primal and the dual
problem is $0$. Then, the condition $W_p^p(\br,\bs) = \scp{\br}{\bu} +
\scp{\bs}{\bv}$ becomes
$\scp{\br}{\bu + \bv}=0$. Since $u_x + v_{x'} \leq \dpp(x,x')$ for all
$x,x'\in \X$ implies $\bu + \bv
\leq 0$ this yields $\bu = -\bv$. This gives   
\[
    \Phi^*_p(\br, \br) = \left\{ (\bu, -\bu)\in\RR^\X\times\RR^\X : u_x - u_{x'} \leq
    \dpp(x, x'), \:
x,x'\in\X \right\}
\]
and the following more compact
representation of the dual solutions in the case $\br=\bs$, independent of
$\br$:
\begin{equation}
\Phi_p^*(\br,\br) = \Phi^*_p \times \left( -\Phi_p^* \right).
\label{eq:null_Phi*}
\end{equation}
\subsection{Proof of Theorem \ref{thm:full}}
\begin{enumerate}[a)]
  \item With the notation introduced in Theorem \ref{thm:full}, $n\brh_n$ is a
    sample of size $n$ from a multinomial distribution with
    probabilities $\br$. Therefore, $\sqrt{n}(\brh_n - \br)\Rightarrow
    \bm{G}$ as $n\ra \infty$ {\citep[Thm. 14.6]{wasserman_all_2011}}. The
    Hadamard derivative of the
    map $(\br, \bs) \mapsto W_p^p(\br,\bs)$ as given in Theorem
    \ref{thm:derivative_Wasserstein} can now be
    used in the delta-method
    from Theorem \ref{thm:delta_method}. Together with the representation \eqref{eq:null_Phi*} of
    the set of dual solutions $\Phi^*_p(\br, \bs)$, this yields
    \[
      \sqrt{n}W_p^p(\brh_n, \br) \Rightarrow \max_{(\bu, \bv)\in
      \Phi^*_p(\br, \br)} -\scp{\bu}{\bm{G}} \:\simD\: \max_{\bu\in\Phi^*_p}
      \scp{\bu}{\bm{G}}.
    \]
    Here and in the following $Z_1\simD Z_2$ means the distributional equality of the
    random variables $Z_1$ and $Z_2$.
    Applying to this the Continuous Mapping Theorem with the map $t\mapsto
    t^{1/p}$ gives the assertion.
  \item Consider the map
    $(\br, \bs)\mapsto W_p(\br, \bs) = (W_p^p(\br, \bs))^{1/p}$. By Theorem
    \ref{thm:derivative_Wasserstein} and the chain rule for
    Hadamard directional derivatives {\citep[Prop. 3.6]{Shapiro1990}}, the Hadamard derivative of
    this map at $(\br, \bs)$  is
    given by
    \begin{equation}
      (\bh_1, \bh_2) \mapsto  p^{-1}W_p^{1-p}(\br, \bs)\left\{ \max_{(\bu,
      \bv)\in\Phi^*_p(\br,\bs)}
      -(\scp{\bu}{\bh_1} + \scp{\bv}{\bh_2}) \right\}.
      \label{eq:der_alt}
    \end{equation}
    An application of the delta-method of Theorem \ref{thm:delta_method} concludes this part.
  \item and d). Note that under the assumptions of the Theorem 
    \begin{equation}
      \sqrt{\frac{nm}{n + m}}\left( (\brh_n, \bsh_m) - (\br, \bs)  \right)
      \Rightarrow (\sqrt{\lambda} \bm{G}, \sqrt{1 - \lambda}\bm{H}).
      \label{eq:conv_multi_twosample}
    \end{equation}
    Part d) follows with the delta-method from \eqref{eq:der_alt} and
    \eqref{eq:conv_multi_twosample}.

    For part c) we use, as we did for a), the derivative given in Theorem
    \ref{thm:derivative_Wasserstein} and the Continuous
    Mapping Theorem.
    The limit distribution is 
    \[
      \left\{ \max_{(\bu,\bv)\in\Phi^*_p(\br, \bs)} (\sqrt{\lambda}\scp{\bm{G}}{\bu} +
      \sqrt{1 - \lambda}\scp{\bm{H}}{\bv}) \right\}^{1/p}.
    \]
    Note that if $\br = \bs$
    we have $(\bu,\bv)\in\Phi^*_p(\br, \bs)$ if and only if $\bu\in\Phi^*_p$ and
    $\bv = -\bu$, by \eqref{eq:null_Phi*} and \eqref{eq:sets}. Hence, with
    $\bm{G} \simD \bm{H}$ we conclude
    \[
      \begin{split}
        \max_{(\bu,\bv)\in\Phi^*_p(\br, \bs)} (\sqrt{\lambda}\scp{\bm{G}}{\bu} +
        \sqrt{1 - \lambda}\scp{\bm{H}}{\bv}) 
        & \simD \max_{\bu\in\Phi^*_p} (\sqrt{\lambda}\scp{\bm{G}}{\bu} -
        \sqrt{1 - \lambda}\scp{\bm{H}}{\bu}) \\
        & \simD  
         \max_{\bu\in\Phi^*_p} \sqrt{\lambda  + (1 -
        \lambda)}\scp{\bm{G}}{\bu} \\
        & = \max_{\bu\in\Phi^*_p} 
        \scp{\bm{G}}{\bu}.
      \end{split}
    \]
\end{enumerate}

\subsection{Explicit limiting distribution for tree metrics}
\label{sec:trees}
Assume that the metric structure on $\X$ is given by a weighted tree, that is, an
undirected connected graph $\T = (\X, E)$
with vertices  $\X$ and edges $E \subset \X\times \X$ that
contains no cycles. We assume the edges to be weighted by a function
$ w:E \ra \RR_{>0}$.
For $x, x'\in \X$ let $e_1,\dots,e_l\in
E$ be the unique path in $\T$ joining $x$ and $x'$, then the length of this
path, 
$d_\T(x,x') = \sum_{j=1}^l w(e_j)$ 
defines a metric $d_\T$ on $\X$.
Without imposing any further restriction on $\T$, we assume it to be rooted at
$\root(\T)\in \X$, say. Then, for $x\in \X$ and $x\neq\root(\T)$  we may define
$\parent(x)\in \X$ as the
immediate neighbor of $x$ in the unique path connecting $x$ and $\root(\T)$. 
We set $\parent(\root(\T))=\root(\T)$.
We also define $\children(x)$ as the set of vertices $x'\in \X$ such that there
exists a sequence $x' = x_1, \dots , x_l = x \in \X$ with $\parent(x_j) =
x_{j+1}$ for $j=1,\dots,l-1$. Note that with this definition $x\in\children(x)$. Additionally,  
define the  linear operator $S_\T : \RR^\X \ra \RR^\X$ 
\[
  (S_\T \bu)_x = \sum_{x'\in\children(x)} u_{x'}.
\]
\begin{thm}
  \label{thm:trees}
  Let $p\geq 1$, $\br\in \mathcal{P}_\X$, defining a probability distribution on
  $\X$ and
  let the empirical measures  $\brh_n$ and
  $\bsh_m$
  be generated by independent random variables $X_1,\dots, X_n $ and $Y_1 ,
  \dots Y_m$, respectively, all drawn from
  $\br = \bs$. 

  Then, with a Gaussian vector $\bm{G}\sim\mathcal{N}(0, \Sigma(\br))$ as
  defined in
  \eqref{eq:def_Sigma} we have the following.
  \begin{enumerate}[a)]
    \item \textbf{(One sample)} As $n\ra\infty$, 
      \begin{equation}
        n^{\frac{1}{2p}} W_p(\brh_n, \br) \Rightarrow \left\{\sum_{
        x \in \X} |(S_\T \bm{G})_x| d_\T(x,\parent(x))^p\right\}^{\frac{1}{p}}
        \label{eq:weak_conv_trees}
      \end{equation}
    \item \textbf{(Two samples)} If  $n\wedge m\ra\infty$ and
    $n/(n + m)\ra\lambda\in(0,1)$ we have 
      \begin{equation}
        \left( \frac{nm}{n+m} \right)^{\frac{1}{2p}} W_p(\brh_n, \bsh_m)
        \Rightarrow \left\{
          \sum_{ x \in \X} |(S_\T \bm{G})_x|
        d_\T(x,\parent(x))^p\right\}^{\frac{1}{p}}.
      \end{equation}
  \end{enumerate}
\end{thm}
The proof of Theorem \ref{thm:trees} is given in the supplementary material. The
theorem includes the special case of a discrete measure on the
real line, that is $\X\subset \RR$, since in this case, $\X$ can be regarded as
a simple rooted tree consisting of only one branch. 
\begin{cor}[{\citealp[Theorem 2.6]{Samworth2004}}]
  \label{cor:samworth}
  Let $\X = \{x_1 < \dots < x_N\}\in\RR$, $\br\in\mathcal{P}_\X$ and $\brh_n$ the
  empirical measure generated by i.i.d. random variables $X_1, \dots ,
  X_n\sim\br$. With $\bar{r}_j = \sum_{i=1}^j r_{x_i} $, for $j=1,\dots N$ and
  $B$ a standard Brownian bridge, we have as $n\ra\infty$, 
  \begin{equation}
    n^{\frac{1}{4}} W_2(\brh_n, \br) \Rightarrow \left\{ \sum_{j=1}^{N-1}
    |B(\bar{r}_j)|
    (x_{j+1} - x_j)^2 \right\}^{\frac{1}{2}}.
    \label{eq:samworth}
  \end{equation}
\end{cor}

%% file: source/appl.tex
\section{Simulations and applications}
\label{sec:appl}
The following numerical experiments were performed using R
\citep{good_r_core_team_r:_2016}.
All computations of Wasserstein distances and optimal transport plans as well as
their visualizations were performed with the R-package \texttt{transport}
\citep{Schuhmacher2014,Gottschlich2014}.
The code used for the computation of the limiting distributions is available as
an R-package \texttt{otinference} \citep{sommerfeld_otinference:_2017}.
\subsection{Speed of convergence}
We investigate the speed of convergence to the limiting distribution in Theorem
\ref{thm:full} in the one-sample case under the null hypothesis. To this end, we
consider as ground space $\X$ a regular two-dimensional $L\times L$ grid with
the
euclidean distance as the metric $d$ and $L= 3, 5, 10$. We generate five random measures
$\br$ on $\X$ as realizations of a Dirichlet random variable with concentration
parameter $\bm{\alpha} = (\alpha, \dots, \alpha)\in\RR^{L\times L}$ for
$\alpha = 1, 5, 10$. Note, that $\alpha = 1$ corresponds to a uniform
distribution on the probability simplex.
For each measure, we generate $20,000$ realizations of $n^{1/2p}W_p(\brh_n,
\br)$ with $n \brh_n\sim\mathrm{Multinom}(\br)$ for
$n =  10, 1000, 1000$ and of the theoretical
limiting distribution given in Theorem  \ref{thm:full}. The Kolmogorov-Smirnov
distance (that is, the maximum absolute difference between their cdfs) between
these two samples (averaged over the five measures) is shown in Figure
\ref{fig:conSpeed}.
\begin{figure}
\thisfloatpagestyle{empty}
  \centering
  \includegraphics[width=0.45\textwidth]{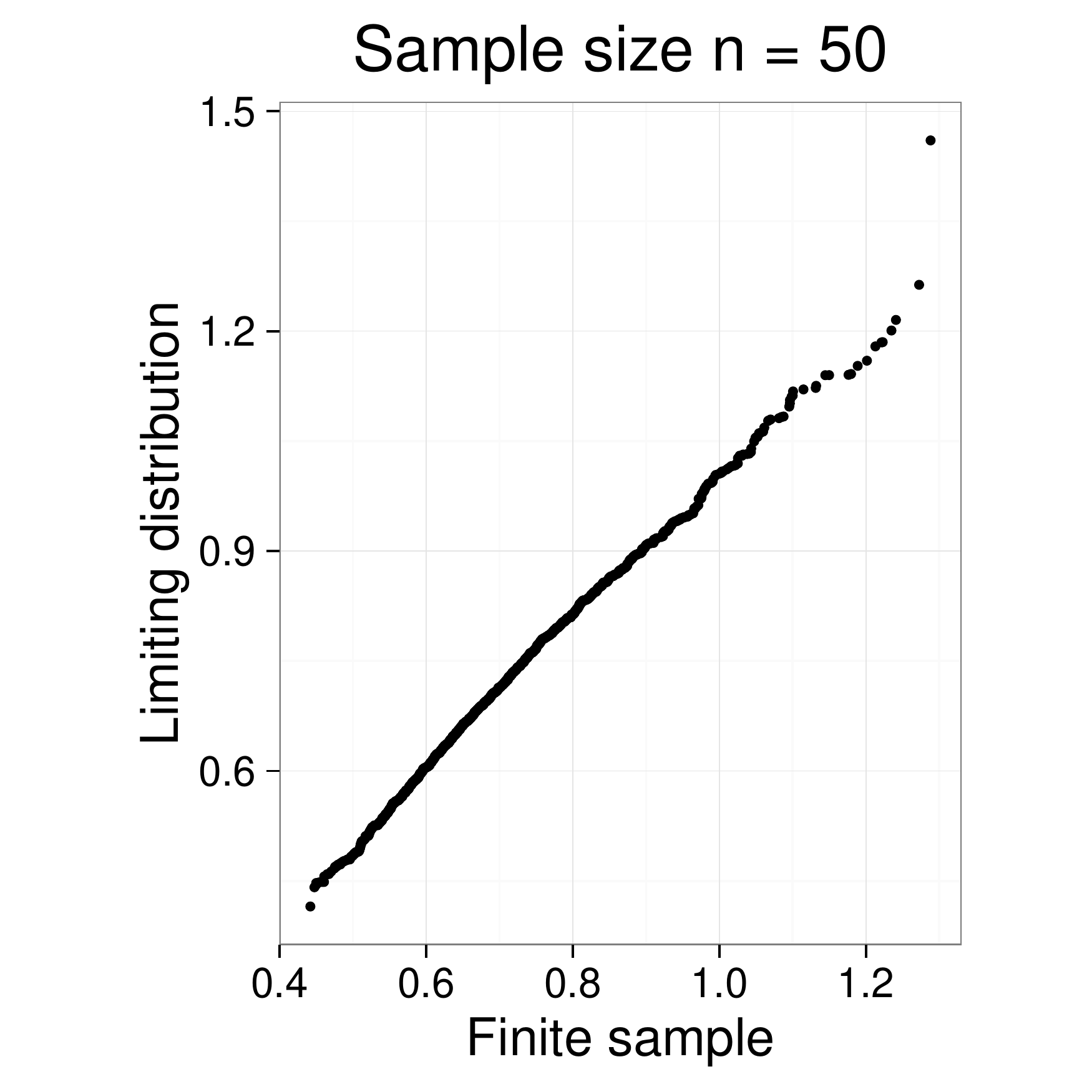}
  \hfill
  \includegraphics[width=0.45\textwidth]{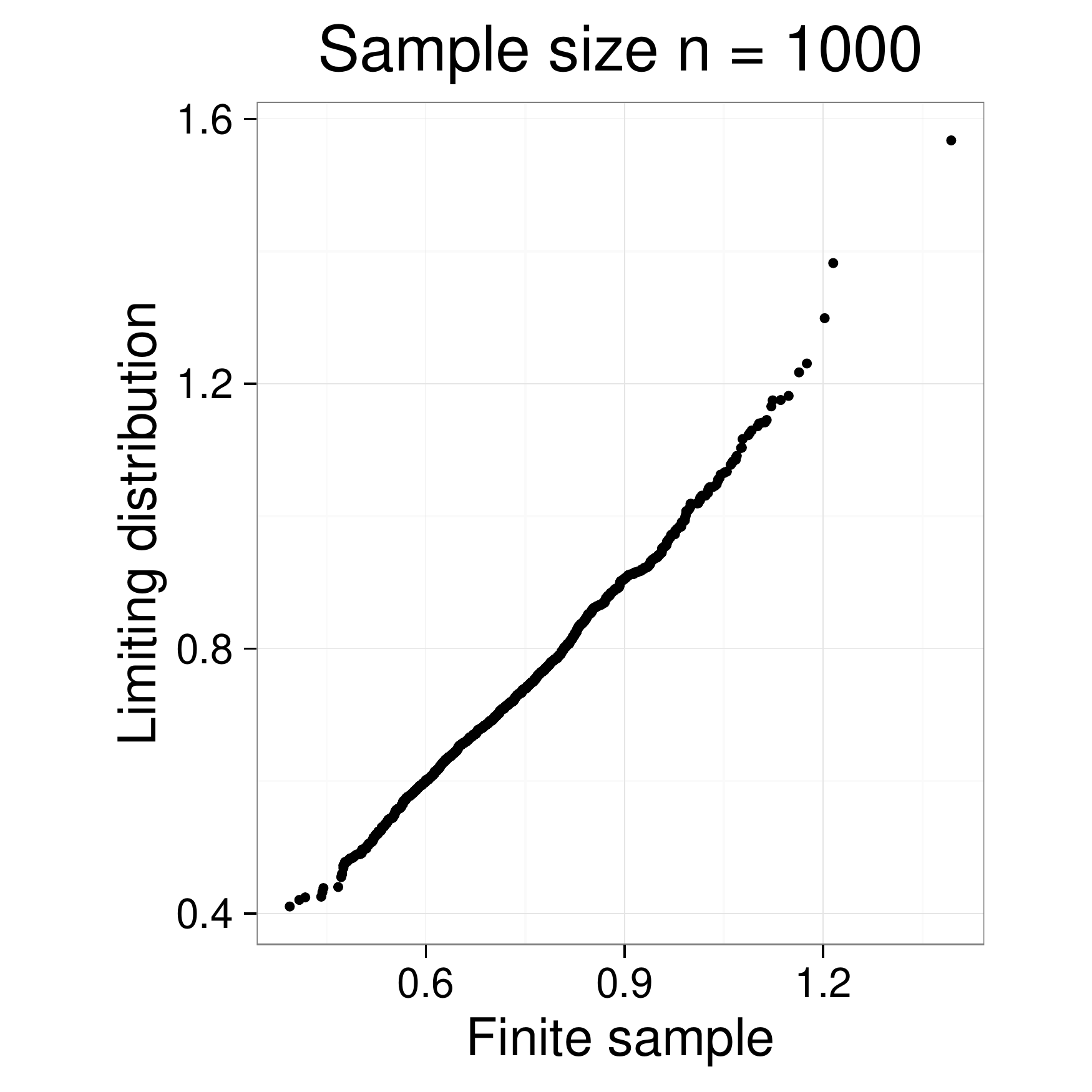}\\
  \vfill
  \includegraphics[width=0.49\textwidth]{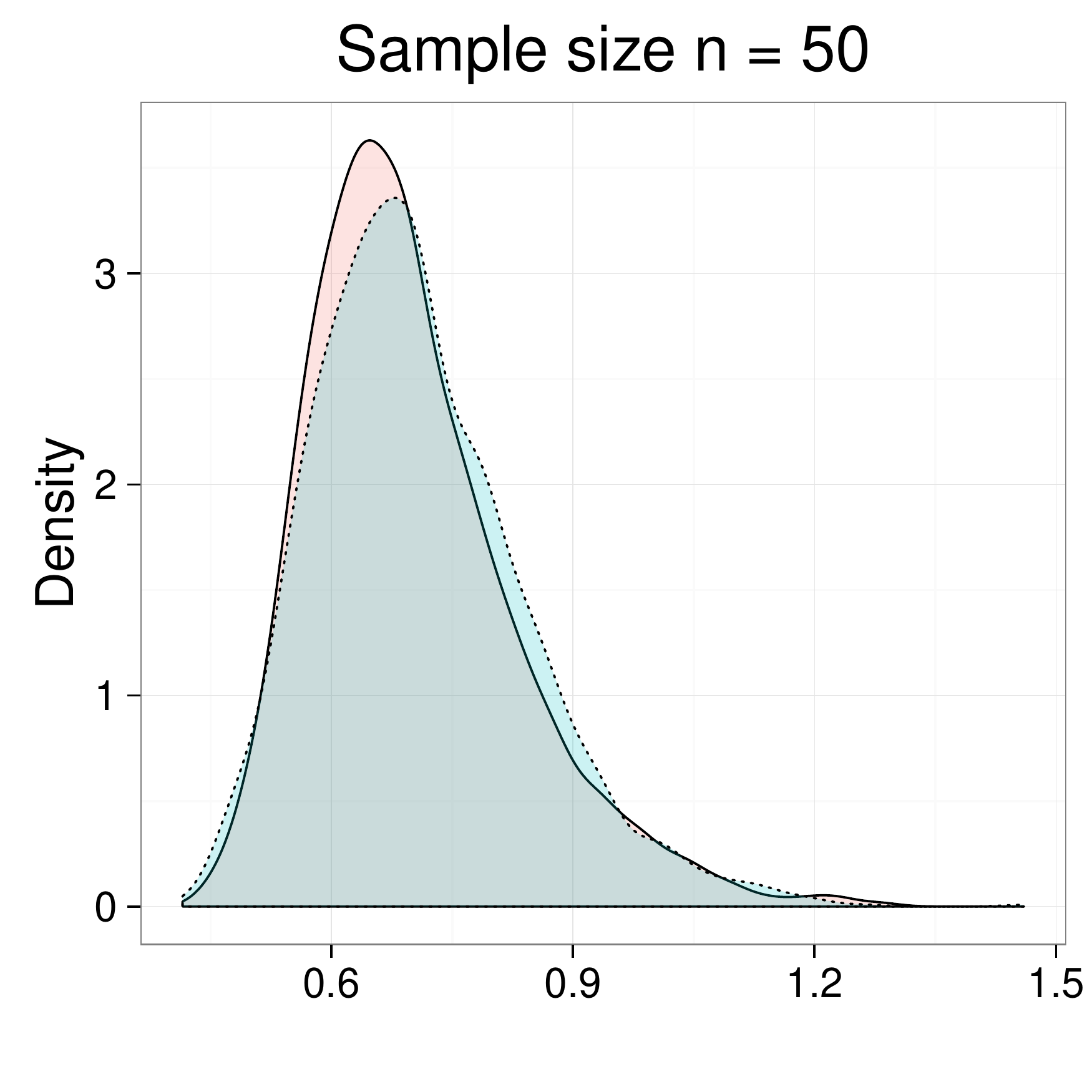}
  \hfill 
  \includegraphics[width=0.49\textwidth]{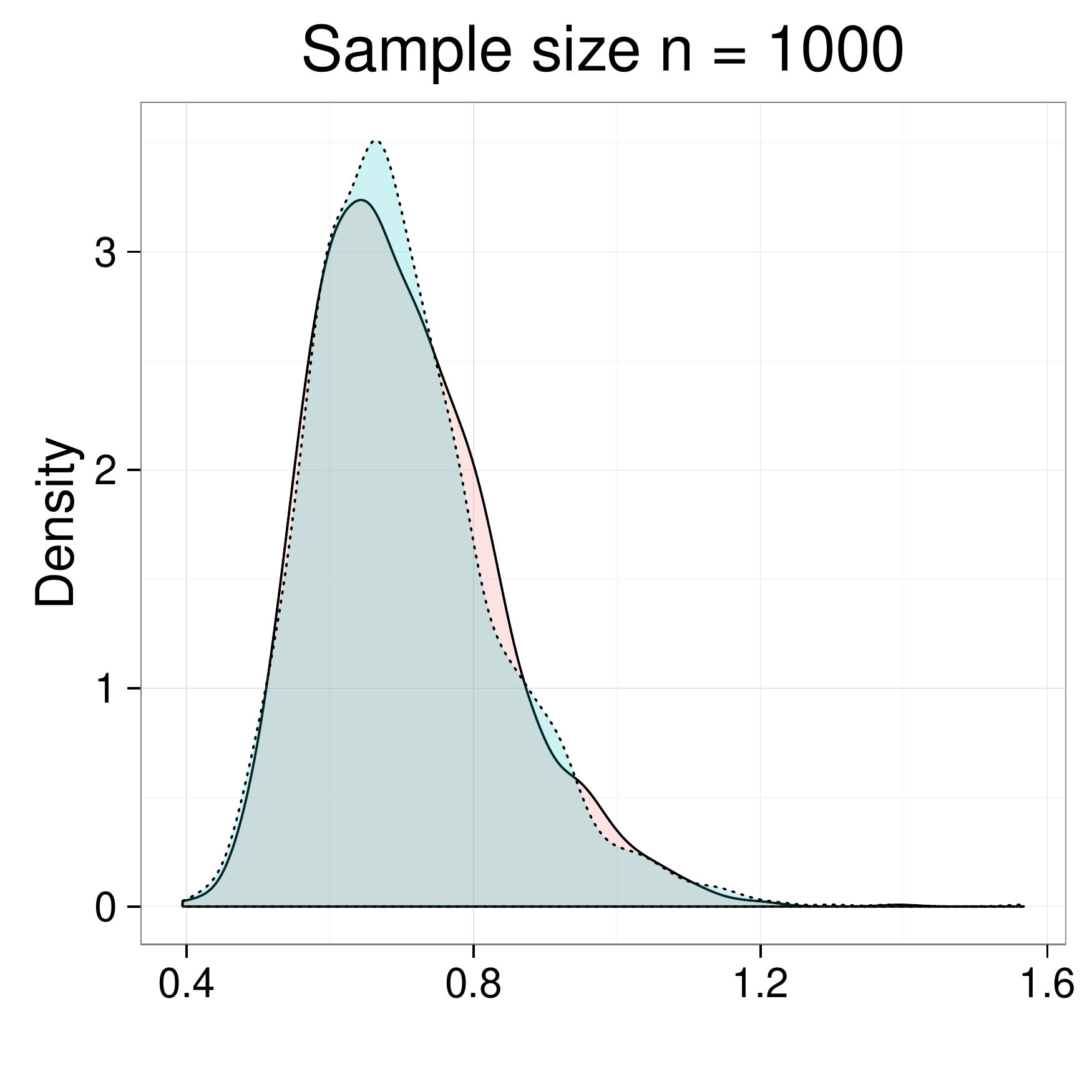}\\
  \vfill
  \includegraphics[width=0.95\textwidth]{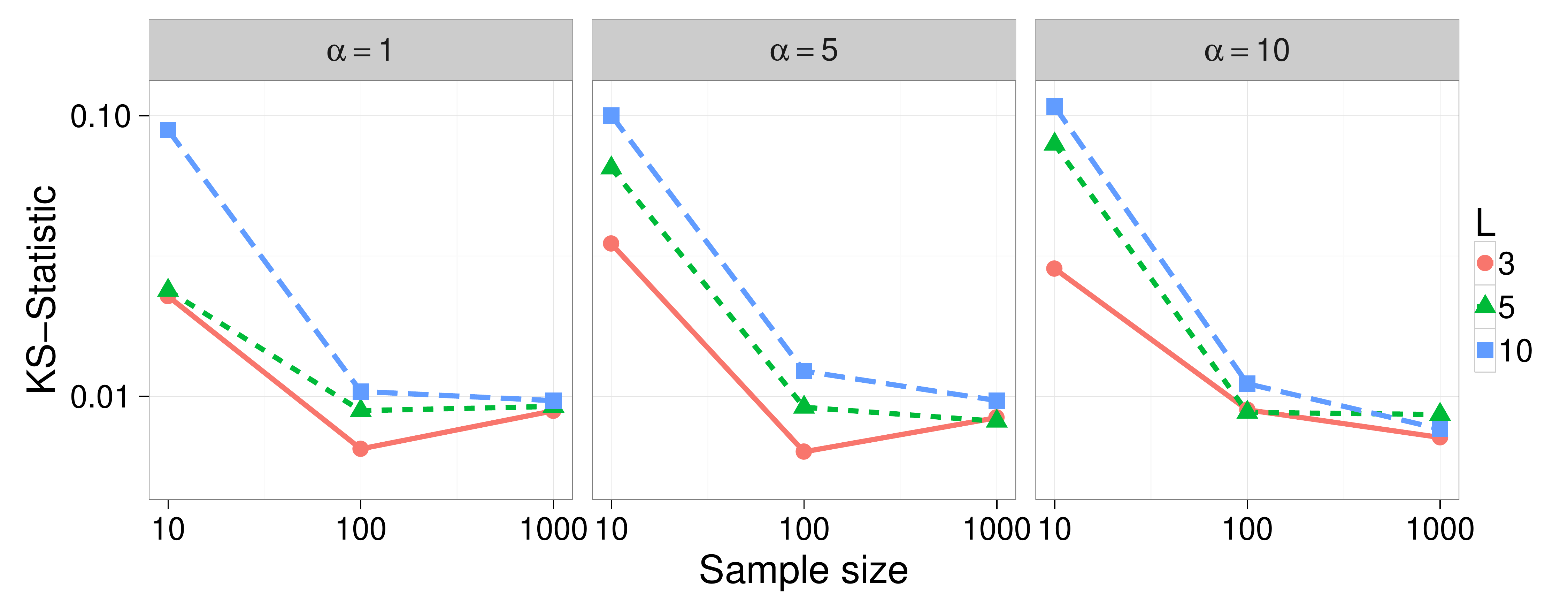}
  \caption{Comparison of the finite sample distribution and the theoretical
  limiting distribution on a regular grid of length $L$ for different sample
sizes. The two top rows show Q-Q-plots and kernel density estimates (bandwidth:
Silverman's rule of thumb \citep{silverman_density_1986}, solid line: finite
sample, dotted line: limiting distribution) for $L = 10$. Last row shows the KS
statistic between the two distributions as a function of the sample size for
different $L$ and for different concentration parameters $\alpha$.}
  \label{fig:conSpeed}
\end{figure}
The experiment shows that the limiting distribution is a good approximation of
the finite sample version even for small sample sizes. For the considered
parameters the size of the ground space $N=L^2$ seems to slow the convergence only
marginally. Similarly, the underlying measure seems to have no sizeable effect
on the convergence speed as the dependence on the concentration parameter
$\alpha$ demonstrates.
\subsection{Testing the null: real and synthetic fingerprints}
\label{sub:FP}
The generation and recognition of synthetic fingerprints is a topic of great
interest in forensic science and current state-of-the-art methods \citep{cappelli_synthetic_2000} produce
synthetic
fingerprints that even human experts fail to recognize as such
{\citep[p. 292ff]{maltoni_handbook_2009}}. Recently,
\cite{gottschlich_separating_2014} presented a method using the Wasserstein
distance that is able to distinguish synthetic from real fingerprints with high
accuracy. Their method is probabilistic in nature, since it is based on a
hypothesized unknown distribution of certain features of the fingerprint. We use
our distributional limits to assess the statistical significance of the
differences.

\paragraph{Minutiae histograms} The basis for the comparison of fingerprints are
so called minutiae which are key qualities in biometric identification based on
fingerprints \citep{jain_technology:_2007}. They are certain characteristic
features such as bifurcations of the line patterns of the
fingerprint. Each of the minutiae have
a location in the fingerprint and a direction such that it can be
characterized by two real numbers and an angle. Figure \ref{fig:minutiae} shows
a real and a synthetic fingerprint with their minutiae.

The recognition method of \cite{gottschlich_separating_2014} considers pairs of
minutiae and records their distance and the difference between their angles.
Based on these two values each minutiae pair is put in one of 100 bins arranged
in a regular grid (10 directional by 10 distance bins) to obtain a so called
minutiae histogram (MH). Based on the bin-wise mean of MHs for several
fingerprints to construct a typical MH, they found that  the
proximity in Wasserstein distance  to these references is a
good classifier for distinguishing real and synthetic fingerprints.

In order to assess the statistical significance of the difference in minutiae
pair distributions,
we consider fingerprints from the databases 1 and 4 of the
Fingerprint
Verification Competition of 2002 \citep{maio_fvc2002}, containing
$110$ real and
synthetic fingerprints, respectively. From each database the minutiae were
obtained by
automatic procedure  using a commercial off-the-shelf program. For each
fingerprint
we chose disjoint minutiae pairs at random  to avoid the issue of pairs being
dependent yielding a total of 1917 and 1437 minutiae pairs from real and
synthetic fingerprints, respectively.
\begin{figure}
  \centering
  \includegraphics[width=0.7\textwidth]{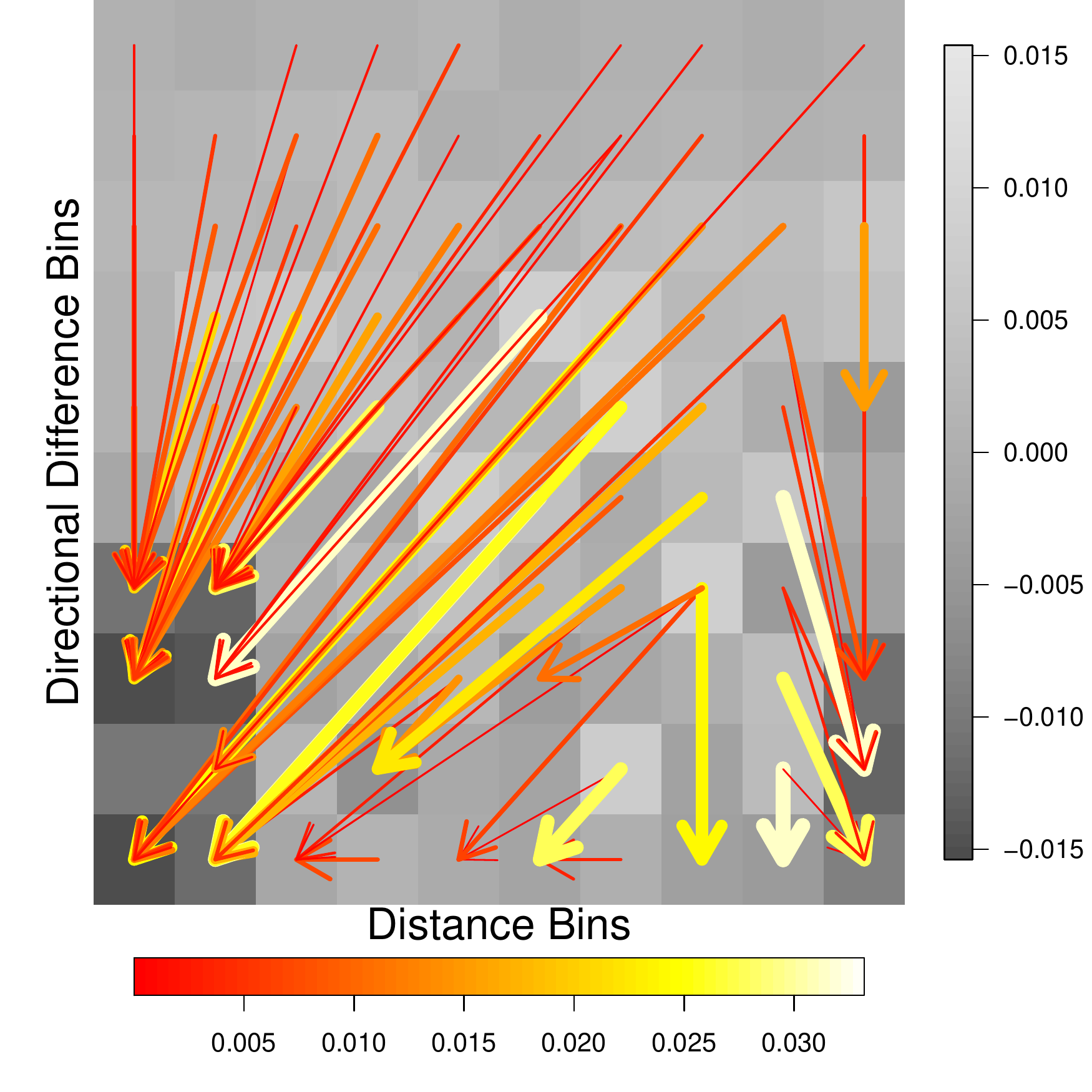}
  \caption{The optimal transport plan between the MHs of real and fake
  fingerprints. The grey values indicate the magnitude of the difference of the
  two MHs. The arrows show the transport. The amount of mass transported is
encoded in the color and thickness of the arrows.}
  \label{fig:FPtransport}
\end{figure}

While two-sample tests for univariate data are  abundant and well studied there
are no multivariate methods that could be considered standard in this setting. 
Therefore, we report on the findings of several tests from the literature for
comparison with the Wasserstein based method from \eqref{eq:two_sample_null}.
We tested the null hypothesis of the underlying distributions being equal 
for the un-centered, the centered and the centered and scaled (to variance
$1$) data to assess effects beyond first
moments using the following methods: 1) comparing the empirical Wasserstein
distance $W_1$ after binning on a regular $10\times 10$ grid with the limiting
distribution from Theorem \ref{thm:full}; 2) a permutation test; 3) the
crossmatch test proposed by \citep{rosenbaum_exact_2005} and 4) the kernel based
test \citep{anderson_two-sample_1994}  implemented in the R package \textup{ks}.

Table
\ref{tab:FPResults} shows
the resulting empirical distributions on a $10\times 10$ grid and the $p$-values
for the different tests.
\begin{table} 
\caption{\label{tab:FPResults} Results of different two-sample tests for difference in the
distribution of MHs of real and fake fingerprints.}
        \begin{tabular}{rrrrr}
          & Wasserstein & Crossmatch & Permutation & KDE \\ 
          \hline
          Raw & 0.00E+00 & 2.99E-01 & 1.00E-03 & 1.12E-08 \\ 
          Centered & 4.00E-04 & 4.48E-05 & 1.00E-03 & 2.60E-21 \\ 
          Centered \& Scaled & 2.54E-02 & 1.01E-02 & 1.71E-01 & 1.79E-14 \\ 
        \end{tabular}
\end{table}
The  differences are highly significant according to all tests, except the
permutation test for the centered and scaled data. In this  particular example  at least,
the Wasserstein based
test seems to be able to pick up differences in distributions (in the first
moment and beyond) at least as good as current state-of-the-art methods.  

In addition to testing, the Wasserstein method provides us with an optimal
transport plan, transforming one measure into the other. For the minutiae
histograms under consideration this is illustrated in Figure
\ref{fig:FPtransport}. This transport plan gives information beyond a simple
test for equality as it highlights structural changes in the distribution.
In this specific application it reveals how in the minutiae histogram of
synthetic fingerprints compared to the one of real fingerprints mass has been
shifted from large and intermediate directional differences to smaller ones. In
particular to small and large distances, and only to a lesser extent to
intermediate distances. In conclusion one may say that synthetic fingerprints
show smaller differences in the directions of minutiae and stronger clustering
of minutiae distances around small and large values.
Insight of this sort may lead to improved generation or detection of synthetic
fingerprints.
\begin{figure}
  \centering
  \includegraphics[width=0.49\textwidth,height=0.18\textheight]{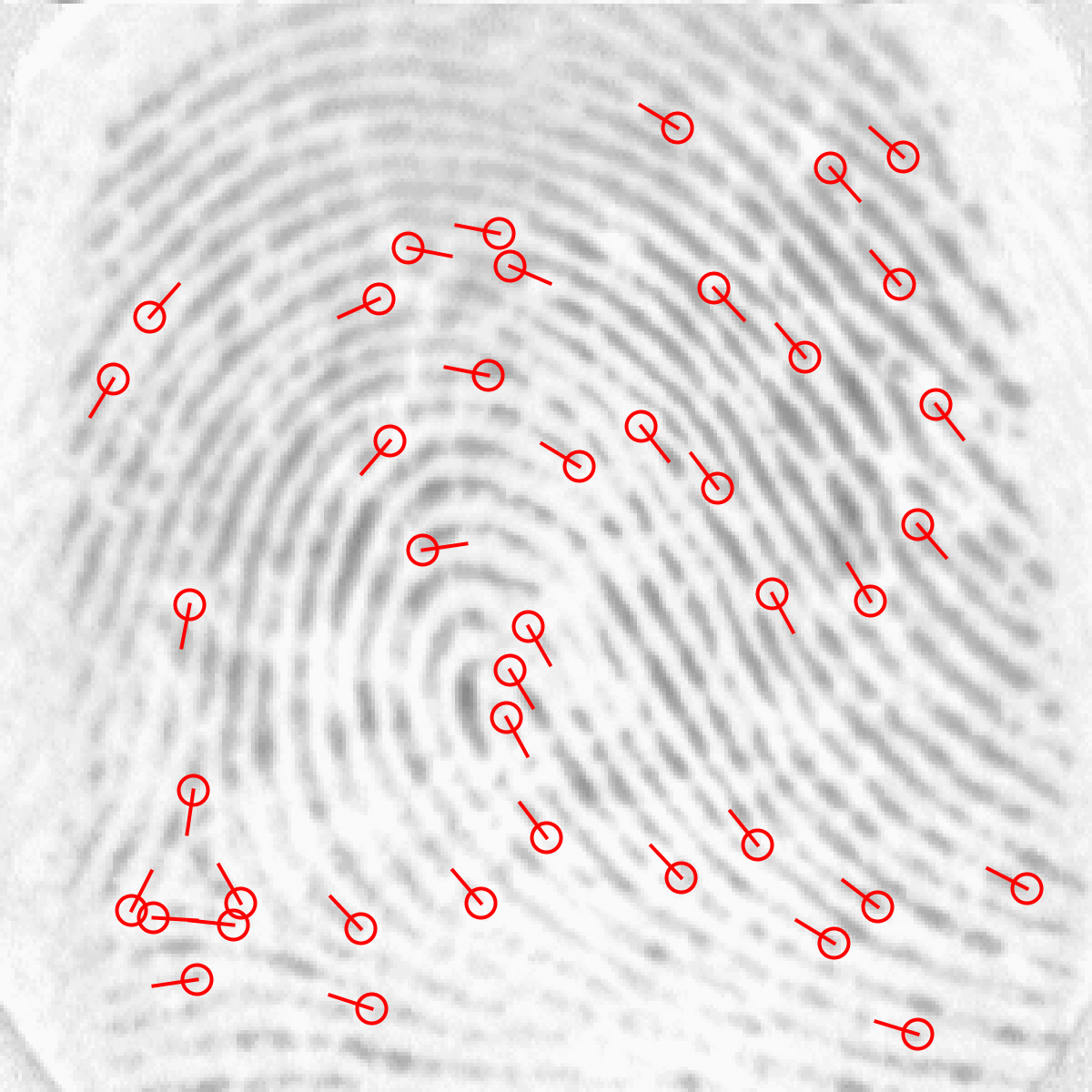}
  \includegraphics[width=0.49\textwidth,height=0.18\textheight]{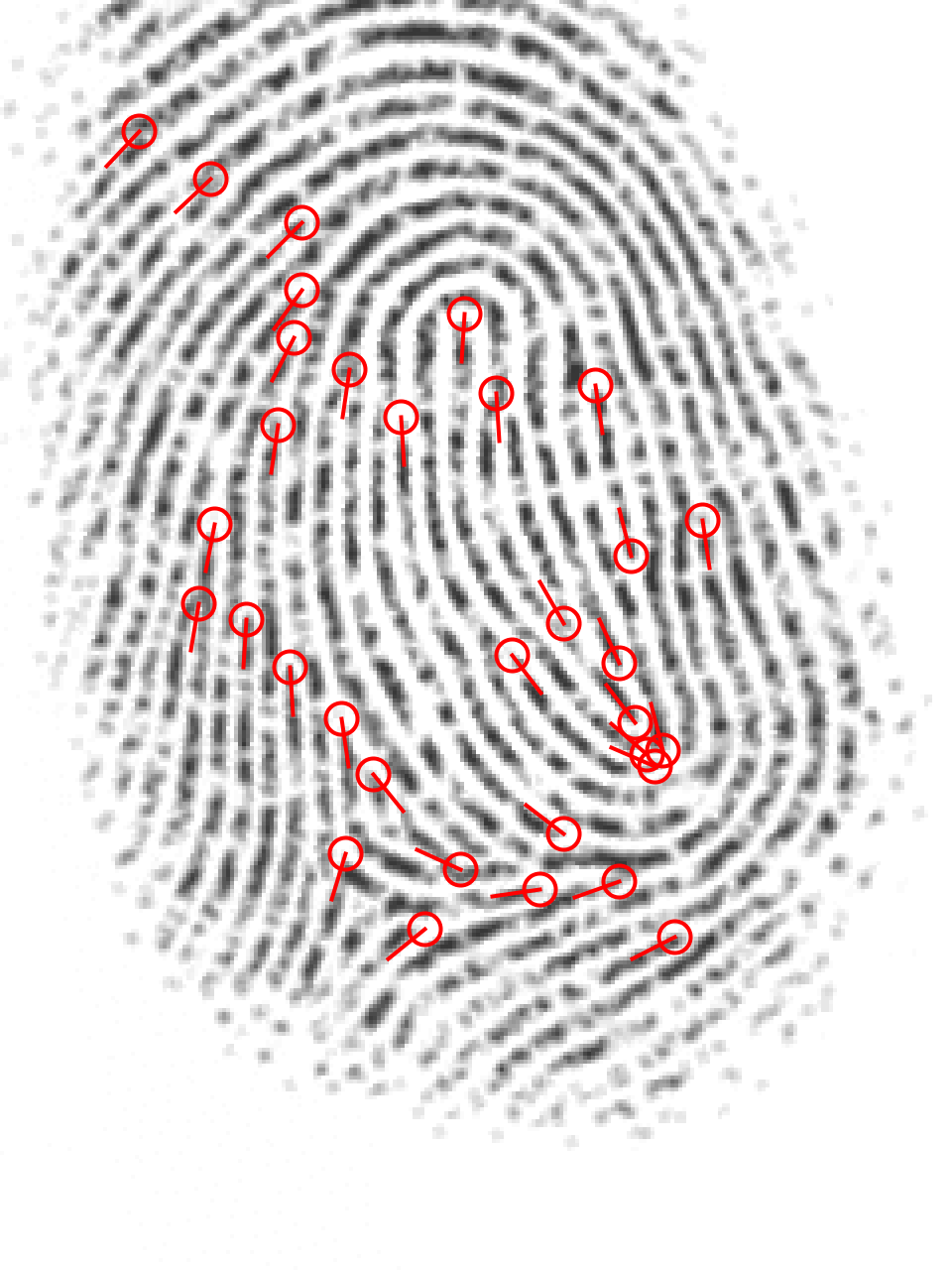}
  \\
    \includegraphics[width=0.95\textwidth]{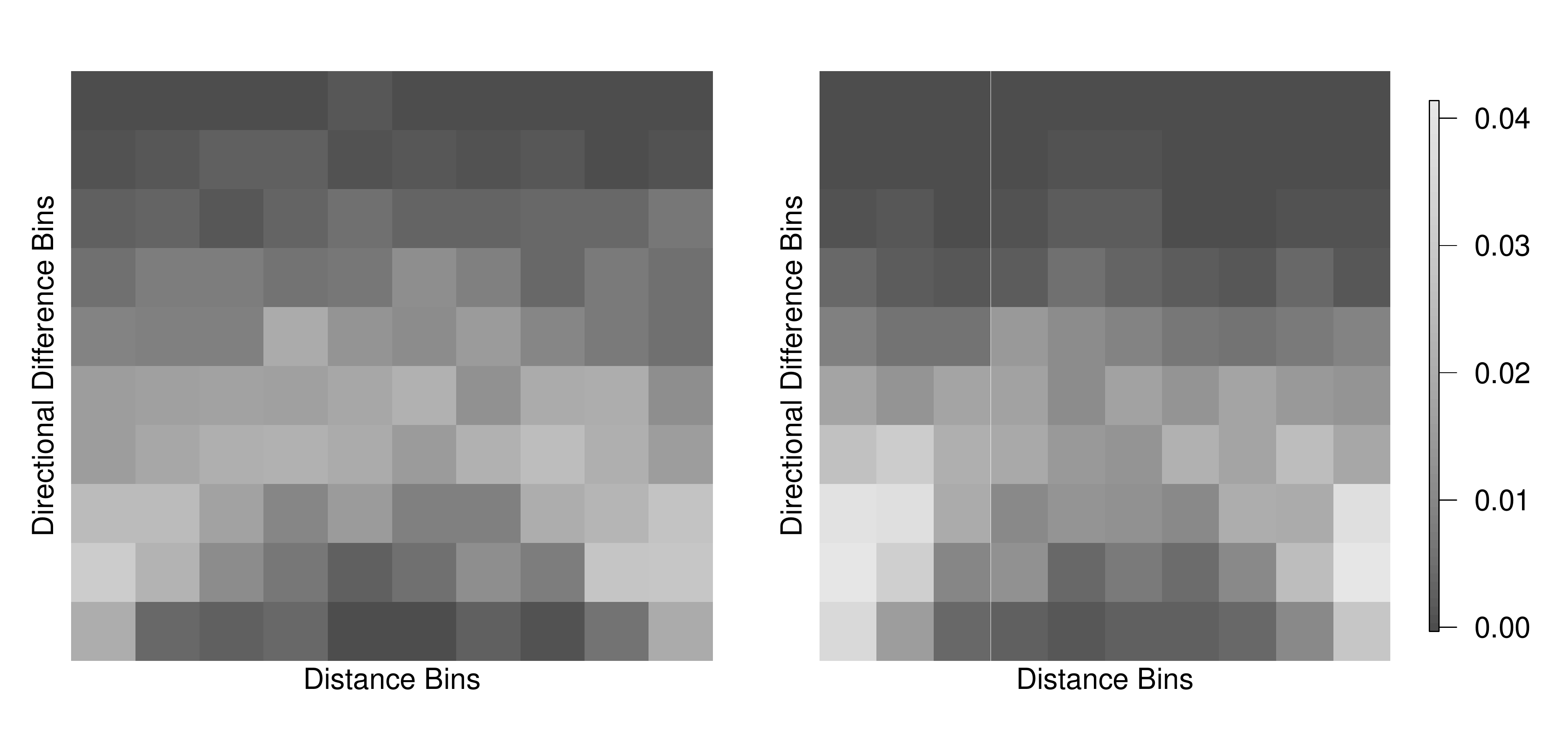}\hfill{}
  \caption{Top row: Minutiae of a real (left) and a synthetic (right)
  fingerprint. Bottom row: Minutiae histograms of real and synthetic
fingerprints.}
  \label{fig:minutiae}
\end{figure}

\subsection{Asymptotic under the alternative: metagenomics}
\label{sub:appl_alt}
\begin{figure} 
  \includegraphics[width=0.66\textwidth]{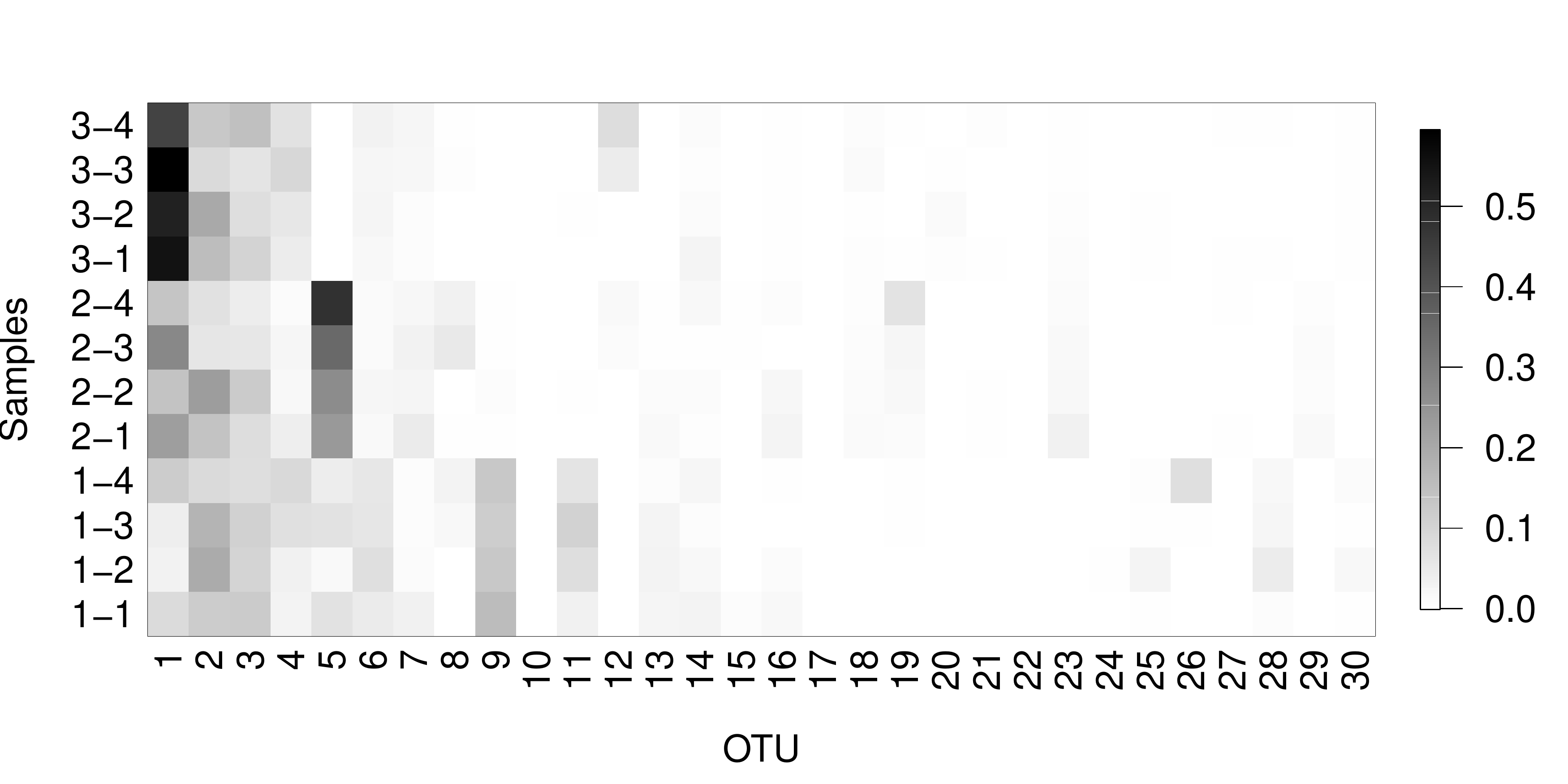}
  \includegraphics[width=0.33\textwidth]{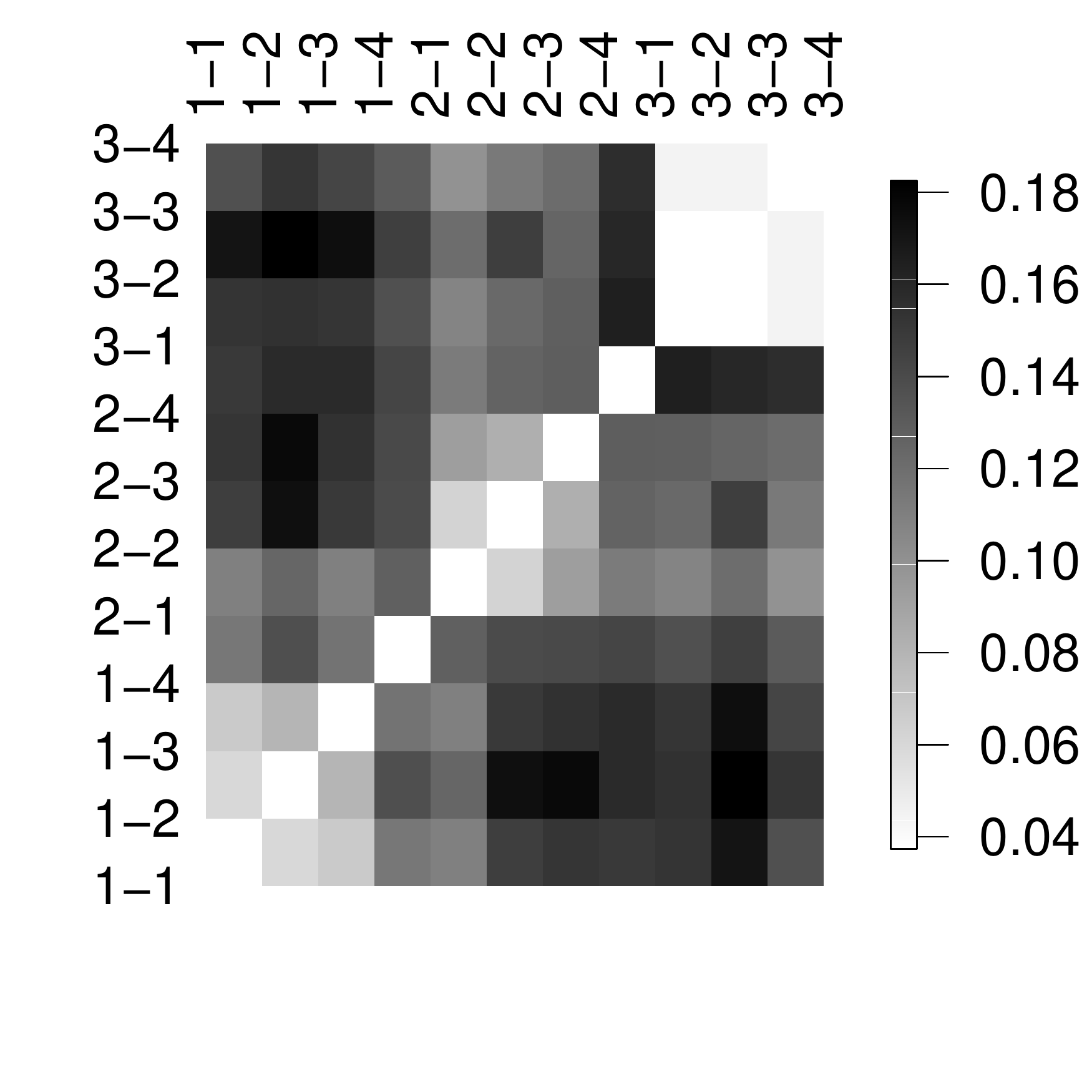}
    \caption{Relative abundances of the $30$ first OTUs in the $12$ samples
      (left) and
     Wasserstein distances of the  microbial communities (right). Here, $ij$ is
   the $j$-th sample of the $i$-th person.}
  \label{fig:stoolResults} 
\end{figure}
Metagenomics studies microbial communities by analyzing genetic
material in an environmental sample such as a stool sample of a human.
High-throughput sequencing techniques no longer require cultivated cloned
microbial cultures  to perform sequencing. Instead, a sample with potentially
many different species can be analyzed directly and the abundance of each
species in the sample can be recovered.
The applications of this technique are countless and constantly growing.
In particular, the composition of microbial communities in the human gut has
been associated with obesity, inflammatory bowel disease and others
\citep{turnbaugh_human_2007}.

The analysis of  a sample with high-throughput sequencing techniques yields
several thousands to many hundreds of thousands sequences. After elaborate
pre-processing, these sequences are aligned to a reference database and
clustered in \textit{operational taxonomic units} (OTUs). These OTUs can be
thought of (albeit omitting some biological detail) as the different species present
in the sample. 
For each OTU this analysis yields the number of sequences
associated with it, that is how often this particular OTU was detected in the
sample. Further, comparing the genetic sequences associated with an OTU yields a
biologically meaningful measure of similarity between OTUs - and hence a
distance. A metagenomic sample can therefore be regarded as a sample in a
discrete metric space with OTUs being the points of the space.
Comparing such samples representing microbial communities is of great interest
\citep{kuczynski_microbial_2010}. The Wasserstein distance has been
recognized to provide valuable insight and to facilitate tests for equality of
two communities \citep{evans_phylogenetic_2012}. This previous application however, relies on a
phylogenetic tree that is build on the OTUs and the distance is then measured in
the tree. This additional pre-processing step involves many parameter choices and
is unnecessary with our method.

A further drawback of the method of \cite{evans_phylogenetic_2012} is that it
only allows for testing
the null hypothesis of two communities being equal. In practice, one frequently
finds that natural variation is so high that even two samples from the same
source taken at different times will be recognized as different. This raises the
question whether variation within  samples from the same source is smaller than
the difference to samples of another source. Statistically speaking we are
looking for confidence sets for differences which are assumed to be different
from zero. This requires asymptotics under the alternative $\br \neq\bs$, which
is provided by Theorem \ref{thm:full}.
\begin{figure}
  \centering
  \includegraphics[width=0.99\textwidth]{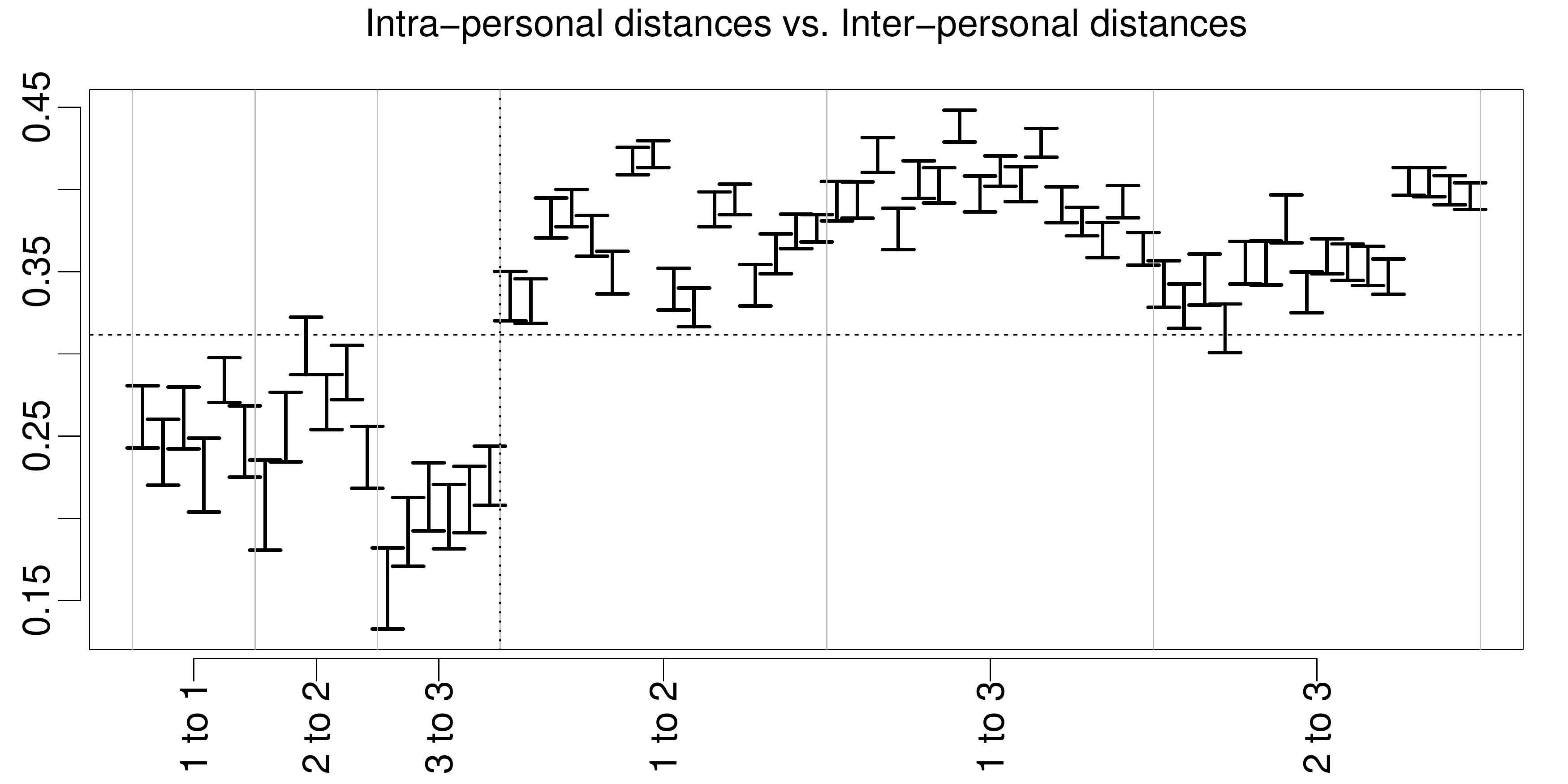}
  \caption{Display of $95\%$ confidence
    intervals of    Wasserstein distances of
  microbial communities. The horizontal axis shows which person pair the
distances belong to (separated by gray vertical lines). The dotted vertical line
separates intra- (left) from inter- (right) -personal distances.}
  \label{fig:stoolIntraInter}
\end{figure}

\paragraph{Data analysis}
We consider part of the data of \cite{costello_bacterial_2009}. Four stool
samples were taken
from each of three persons at different times. We used the preparation of this
data by P. Schloss available at
\url{https://www.mothur.org/w/images/d/d8/CostelloData.zip}. The reads were
pre-processed with the
program \texttt{mothur} \citep{schloss_introducing_2009} using the
procedure outlined in \cite{schloss_reducing_2011} and \cite{454}. 
The relative abundances of the $30$ most frequent OTUs and the Wasserstein-2
distances of the microbial communities are shown in Figure
\ref{fig:stoolResults}. In this and all other figures we use $i-j$ to denote
sample $j$ of person $i$. Note that it is typical for this data that most of the
mass is concentrated on a few OTUs.

The Wasserstein-2 distances  for
all $66$ pairs and their $99\%$ confidence intervals were computed using the
asymptotic distribution in Theorem \ref{thm:full}. The results are shown in
Figure \ref{fig:stoolIntraInter}. The entire analysis took less than a minute on
a standard laptop. The confidence intervals show that intra-personal distances
are in fact significantly smaller than inter-personal distances.

%% file: source/discussion.tex
\section{Discussion}
\label{sec:dis}
We discuss limitations, possible extensions of the presented work and promising directions
for future research.

\paragraph{Beyond finite spaces I: rates in the finite and the continuous
setting ($d=1$)}
The scaling rate in Theorem \ref{thm:full} depends solely on $p$ and is
completely independent of the underlying space $\X$. This contrasts known bounds
on the rate of convergence in the continuous case (see references in the
Introduction),
which exhibit a strong dependence on the dimension of the space and the moments
of the distribution.

Under the null hypothesis (that is, the two underlying population measures are equal) and  when $\X=\RR$ and $p=2$, the scaling rate for a continuous distribution is known
to be $n^{1/2}$, at least under
additional tail conditions (see e.g. \cite{Barrio2005}). This means that
in this case the scaling rate for a discrete distribution is slower
(namely $n^{1/4}$). Under the alternative (different population measures) the scaling rate is $n^{1/2}$ and coincide in the discrete and the continuous
case (see  \cite{Munk1998}).

\paragraph{Beyond finite spaces II: higher dimensions ($d\geq 2$)} 
For a continuous measure $\mu$ the Wasserstein distance is the solution of an
infinite-dimensional optimization problem. Although differentiability results
also exist for such problems (e.g. \cite{shapiro_perturbation_1992}), there are strong indications that the
argument presented here cannot carry over to the this case for $d\geq 2$. This is
most easily seen from the classical results of \cite{ajtai_optimal_1984}. We consider
the uniform distribution on the unit square. For two samples of size $n$ independently drawn
from this distribution, \cite{ajtai_optimal_1984} showed that there
exist constants $C_1,C_2$ such that the 1-Wasserstein distance $\hat{W}_1^{(n)}$ between them
satisfies 
\[
  C_1 n^{-1/2}(\log n)^{1/2} \leq \hat{W}_1^{(n)} \leq C_2 n^{-1/2}(\log n)^{1/2}
\]
with probability $1 - o(1)$. Hence, for $c_n\hat{W}_1^{(n)}$ to have a non-degenerate limit,
we need $c_n = \sqrt{n / \log n}$. However, a common property of all
delta-methods is that they preserve the rate of convergence, which is 
not satisfied here. 
\paragraph{Transport distances on trees}
Complementing our Theorem \ref{thm:trees} a further result on transport distances on trees was proven by
\cite{evans_phylogenetic_2012} in the context of phylogenetic trees for the
comparison of metagenomic samples (see also our application in Section
\ref{sec:appl}). They point out that the Wasserstein-1 distance on trees is
equal to the so-called \textit{weighted uni-frac distance} which is very popular
in genetics. Inspired by this distance they give a formal generalization
mimicking a cost exponent $p>1$ and consider its asymptotic behavior. However,
as they remark, these generalized
expressions are no longer related (beyond a formal resemblance) to Wasserstein
distances with cost exponent $p>1$. Comparing the performance of their ad-hoc
metric and the true Wasserstein distance on trees that is under consideration
here is an interesting topic for further research. 

\paragraph{Bootstrap}
We showed that while the naive $n$-out-of-$n$ bootstrap fails for the
Wasserstein distance (Section \ref{sub:boot}), the $m$-out-of-$n$ bootstrap is
consistent. An interesting and challenging question is how $m$ should
be chosen.

\paragraph{Wasserstein barycenters}
Barycenters in the Wasserstein space \citep{agueh_barycenters_2011} have
recently received much attention {\citep{cuturi_fast_2014,
del_barrio_statistical_2015}. We expect that the techniques developed here can
be of use in providing a rigorous statistical theory (e.g. distributional
limits).
The same applies to geodesic principal component analysis in the Wasserstein
space \citep{bigot_geodesic_2013,seguy_algorithmic_2015}.

\paragraph{Alternative cost matrices and transport distances}
Theorem \ref{thm:full} holds in very large generality for arbitrary cost
matrices, including in particular the case of a cost matrix derived from a
metric but using a cost exponent $p <1$. 

Beyond this obvious modification it seems worthwhile to extend the methodology
of directional
differentiability in conjunction with a delta-method to other functionals
related to optimal transport, e.g. entropically regularized \citep{Cuturi2013a}
or sliced Wasserstein distances \citep{bonneel_sliced_2015}.
This would require a careful investigation of the analytical properties of these
quantities similar to classical results for the Wasserstein distance.
\section*{Acknowledgment}
The authors gratefully acknowledge support by the DFG Research Training Group
2088 Project A1.
They would like to thank L. D\"umbgen, A. Hein, S. Huckemann, C. Gottschlich 
D. Schuhmacher and R. Schultz for helpful discussions and C. Tameling for
careful reading of the manuscript. 


%% file: source/appendix.tex
\appendix
\section{An alternative representation of the limiting distribution}
\label{sec:altalt}
We give a second representation of the limiting distribution under the
alternative $\br\neq \bs$. The random part of the limiting
distribution \eqref{eq:two_sample_alt} is the linear program
\[
  \max_{(\bu, \bv)\in\Phi^*_p(\br, \bs)} \sqrt{\lambda} \scp{\bm{G}}{\bu}  + \sqrt{ 1-
  \lambda} \scp{\bm{H}}{\bv}.
\]
With the representation \eqref{eq:sets} of  $\Phi^*_p(\br, \bs)$ we obtain the dual linear
program
\begin{align*}
  \min\quad &zW_p^p(\br, \bs) + \sum_{x,x'\in \X} w_{x,x'} d^p(x,x') \\
  \text{s.t.}\quad & 
  \bw\geq 0, z\in \RR \\
  & \sum_{x'\in\X} w_{x,x'} +z r_x = G_x \\
  & \sum_{x\in\X} w_{x,x'} +z s_x = H_x \\
\end{align*}
Note that the constraints can only be satisfied if both $\sqrt{\lambda}\bm{G} - z\br$ and
$\sqrt{1-\lambda}\bm{H} - z\bs$ have only non-negative entries and $z\leq 0$. In this case the second term in
the objective function is clearly minimized by $-z\bw^*$, with $\bw^*$ an optimal transport
plan between these two measures $\br - \sqrt{\lambda}\bm{G} / z$ and
$\bs -\sqrt{1-\lambda}\bm{H} / z$ and the second term of the objective function is equal to
$-z W_p^p(\br - \sqrt{\lambda}\bm{G} / z, \bs - \sqrt{1-\lambda}\bm{H} / z)$.

To write this more compactly let us slightly extend our notation. For $\br, \bs
\in\RR^\X$ with $\sum_x r_x = \sum_x s_x = 1$ let 
\[
  \tilde{W}_p^p(\br, \bs) = 
  \begin{cases}
    W_p^p(\br, \bs) & \text{if } \br, \bs \geq 0; \\
    \infty & \text{else.}
  \end{cases}
\]
With this we can thus write the random variable in the limiting distribution
\eqref{eq:two_sample_alt} as the one-dimensional non-linear optimization
problem
\begin{equation}
  \frac{1}{p} W_p^{1-p}(\br, \bs)\min_{z\geq 0} z \left\{\tilde{W}_p^p(\br +
  \sqrt{\lambda}\bm{G} / z, \bs + \sqrt{1-\lambda}\bm{H} / z) - W_p^p(\br, \bs)\right\}.
  \label{eq:altalt}
\end{equation}
\section{Bootstrap}
\label{sub:boot}
In this section we discuss the bootstrap for the Wasserstein distance. In
addressing the usual measurability issues that arise in the formulation of
consistency for the bootstrap, we follow \cite{van_der_vaart_weak_1996}. We denote
by $\brh_n^*$ and $\bsh_m^*$ some bootstrapped versions of $\brh_n$ and
$\bsh_m$. More precisely, let $\brh_n^*$ a measurable function of $X_1
,\dots,X_n$ and random weights $W_1, \dots , W_n$, independent of the data and
analogously for $\bsh_m^*$. This setting is general enough to include many common
bootstrapping schemes.  We say that, with the assumptions and notation of
Theorem \ref{thm:full}, the bootstrap is consistent if 
 the limiting
distribution of 
\[
  \rho_{n,m} \left\{ (\brh_n, \bsh_m) - (\br, \bs)
  \right\} \Rightarrow (\sqrt{\lambda}\bm{G}, \sqrt{1-\lambda}\bm{H}) 
\]
is consistently estimated by the law of 
\[
  \rho_{n,m} \left\{ (\brh_n^*, \bsh_m^*) - (\brh_n,
  \bsh_m) \right\}.
\]
To make this precise, we define for $A\subset\RR^d$, with $d\in\NN$, the set of
bounded Lipschitz-1 functions
\[
  \BL1 (A) = \left\{ f:A\ra\RR: \sup_{x\in A}|f(x)|\leq 1, \quad |f(x_1) -
  f(x_2)| \leq ||x_1 - x_2|| \right\},
\]
where $||\cdot||$ is the Euclidean norm.
We say that the bootstrap versions $(\brh_n^*, \bsh_m^*)$ are consistent if
\begin{equation}
  \begin{split}
    \sup_{f\in\BL1(\RR^\X \times \RR^\X)} |E\left[ f(\rho_{n,m}\left\{ (\brh_n^*,
    \bsh_m^*) - (\brh_n, \bsh_m) \right\}) | X_1, \dots, X_n, Y_1, \dots,Y_m \right] \\
    - E\left[ f( (\sqrt{\lambda}\bm{G}, \sqrt{1-\lambda}\bm{H})) \right]|
  \end{split}
  \label{eq:boot_cons}
\end{equation}
converges to zero in probability.
\paragraph{Bootstrap for directionally differentiable functions}
The most straightforward way to bootstrap $W_p^p(\brh_n,
\bsh_m)$ is to simply plug-in $\brh_n^*$ and $\bsh_m^*$. That is, trying to
approximate the limiting distribution of $\rho_{n,m}\left\{W_p^p(\brh_n, \bsh_m) -
W_p^p(\br, \bs)\right\}$ by the law of 
\begin{equation}
  \rho_{n,m}\left\{ W_p^p(\brh_n^*, \bsh_m^*) - W_p^p(\brh_n, \bsh_m) \right\}
  \label{eq:plug_in_boot}
\end{equation}
conditional on the data.  While for functions that are Hadamard differentiable
this approach yields a consistent bootstrap (e.g. \cite{gill_non-_1989,
van_der_vaart_weak_1996}), it has
been pointed out by \cite{dumbgen_nondifferentiable_1993} and
more recently by \cite{fang_inference_2014} that this is in
general not
true for functions that are only directionally Hadamard differentiable. In
particular the plug-in approach fails for the Wasserstein distance. 

For the Wasserstein distance there are two alternatives. First,
\cite{dumbgen_nondifferentiable_1993} already pointed out that re-sampling fewer
than $n$ (or $m$, respectively) observations yield a consistent bootstrap.
Second, \cite{fang_inference_2014} propose to plug-in $\rho_{n,m}\left\{
(\brh_n^*, \bsh_m^*) - (\brh_n, \bsh_m) \right\}$ into the derivative of the
function.

Recall from Section \ref{sec:asymp_distr} that 
\begin{equation}
  \phi_p:\RR^N\times\RR^N \ra \RR, \quad \phi_p(\bh_1 ,\bh_2)
  = \max_{\bu\in\Phi^*_p}
   \scp{\bu}{\bh_2 - \bh_1}
  \label{eq:der_phi}
\end{equation}
is the directional Hadamard derivative of $(\br,\bs)\mapsto W_p^p(\br,\bs)$ at
$\br=\bs$.
With this notation, the following Theorem summarizes the implications of the
results of \cite{dumbgen_nondifferentiable_1993} and \cite{fang_inference_2014}
for the Wasserstein distance.
\begin{thm}[Prop. 2 of \cite{dumbgen_nondifferentiable_1993} and Thms. 3.2 and 3.3 of \cite{fang_inference_2014}]
  \label{thm:boot}
  Under the assumptions of Theorem \ref{thm:full} let $\brh_n^*$ and $\bsh_m^*$
  be consistent bootstrap versions of $\brh_n$ and $\bsh_m$, that is,
  \eqref{eq:boot_cons} converges to zero in probability. Then, 
  \begin{enumerate}
    \item the plug-in bootstrap \eqref{eq:plug_in_boot} is \textit{not}
      consistent, that is,  
      \begin{equation*}
        \begin{split}
          \sup_{f\in\BL1(\RR)} E\left[ f(\rho_{n,m} \left\{W_p^p(\brh_n^*, \bsh_m^*) -
          W_p^p(\brh_n, \bsh_m) \right\}) | X_1, \dots X_n,Y_1,\dots,Y_m \right] \\
       - E[f(\rho_{n,m}\left\{ W_p^p(\brh_n, \bsh_m) - W_p^p(\br, \bs) \right\})]
     \end{split}
     \end{equation*}
     does \textit{not} converge to zero in probability.
    \item Under the null hypothesis $\br = \bs$, the derivative plug-in
      \begin{equation}
        \phi_p(\rho_{n,m}\left\{ (\brh_n^*,\bsh_m^*) - (\brh_n, \bsh_m) \right\})
        \label{eq:der_boot}
      \end{equation}
      is consistent, that is
      \begin{equation*}
        \begin{split}
          \sup_{f\in\BL1(\RR)}E\left[ f(\phi_p(\rho_{n,m}\left\{ (\brh_n^*,
            \bsh_m^*) - (\brh_n, \bsh_m)
          \right\})) | X_1, \dots,X_n,Y_1, \dots, Y_m  \right] \\
          - E\left[ f\left(\rho_{n,m}\left\{ W_p^p(\brh_n, \bsh_m) - W_p^p(\br,
          \bs) \right\}\right) \right]
        \end{split}
      \end{equation*}
      converges to zero in probability.
  \end{enumerate}
\end{thm}
\section{Proofs}
\subsection{Proof of Theorem \ref{thm:derivative_Wasserstein}}
  By  {\cite[Ch. 3, Thm. 3.1]{Gal1997}} the function $(\br,\bs)\mapsto
  W_p^p(\br,\bs)$ is directionally differentiable with derivative
  \eqref{eq:derivative} in the sense of G\^ateaux, that is, the limit
  \eqref{eq:def_hadamard} exists for a fixed $\bh$ and not a sequence
  $\bh_n\ra\bh$ (see e.g. \cite{Shapiro1990}). To see that this is also a
  directional
  derivative in the Hadamard sense \eqref{eq:def_hadamard} it suffices 
  {\cite[Prop. 3.5]{Shapiro1990}} to show that  $(\br,\bs)\mapsto
  W_p^p(\br,\bs)$ is locally Lipschitz. That is, we need to show that for
  $\br,\br',\bs,\bs'\in\mathcal{P}_\X$ 
  \begin{equation*}
    |W_p^p(\br,\bs) - W_p^p(\br',\bs')| \leq C ||(\br,\bs) - (\br',\bs')||,
  \end{equation*}
  for some constant $C>0$ and some (and hence all) norm $||\cdot||$ on
  $\RR^N\times\RR^N$. Exploiting symmetry, it suffices to show that 
  \begin{equation*}
    W_p^p(\br, \bs) - W_p^p(\br, \bs') \leq C ||\bs - \bs'||
  \end{equation*}
  for some constant $C>0$ and some norm $||\cdot||$. To this end, we employ an
  argument similar to that used to prove the triangle inequality for the
  Wasserstein distance (see e.g. {\cite[p. 94]{villani_optimal_2008}}). Indeed,
  by the gluing Lemma {\cite[Ch. 1]{villani_optimal_2008}} there exist random
  variables $X_1,X_2,X_3$ with marginal distributions $\br,\bs$ and $\bs'$,
  respectively, such that $E[d^p(X_1, X_3)] = W_p^p(\br, \bs')$ and
  $E[d(X_2,X_3)] = W_1(\bs, \bs')$. Then, since $(X_1, X_2)$ has marginals
  $\br$ and $\bs$, we have 
  \begin{align*}
    W_p^p(\br, \bs) - W_p^p(\br, \bs') &\leq E\left[ d^p(X_1, X_2) -
    d^p(X_1, X_3) \right] \\
    & \leq p\diam(\X)^{p-1} E\left[ |d(X_1, X_2) - d(X_1, X_3)| \right]\\
    & \leq p\diam(\X)^{p-1} E\left[ d(X_2, X_3) \right] = p\diam(\X)^{p-1}
    W_1(\bs, \bs') \\ 
    & \leq p\diam(\X)^{p} ||\bs - \bs'||_1, 
  \end{align*}
  where the last inequality follows from {\cite[Thm.
  6.15]{villani_optimal_2008}}. This completes the proof.
  
\subsection{Proof of Theorem \ref{thm:trees}}
\paragraph{Simplify the set of dual solutions $\Phi^*_p$} As a first step, we
rewrite the set of dual solutions $\Phi^*_p$ given in
\eqref{eq:sets} in our tree notation as
\begin{equation}
  \Phi^*_p = \left\{ \bu\in\RR^\X: u_x - u_{x'} \leq
d_\T(x,x')^p, \quad x,x'\in \X \right\}.
\label{eq:Phi*_tree}
  \end{equation}
  The key observation is that in the condition $u_x - u_{x'}\leq d_\T(x,x')^p$ we
  do not need to consider all pairs of vertices $x,x'\in \X$, but only those which
  are joined by an edge. To see this, assume that only the latter condition holds.
  Let $x,x'\in \X$ arbitrary and $x = x_1,
  \dots , x_l = x'$ the sequence of vertices defining the unique path joining
  $x$ and $x'$, such that $(x_j,x_{j+1})\in E$ for $j=1,\dots,n-1$. Then
  \[
    u_x - u_{x'} = \sum_{j=1}^{n-1} (u_{x_j} - u_{x_{j+1}}) \leq \sum_{j=1}^{n-1}
    d_\T(x_j, x_{j+1})^p \leq d_\T(x,x')^p,
  \]
  such that the condition is satisfied for all $x,x'\in \X$. Noting that if two
  vertices are joined by an edge than one has to be the parent of the other, we
  can write the set of dual solutions as
  \begin{equation}
    \Phi^*_p = \left\{ \bu\in \RR^\X  : |u_x -
      u_{\parent(x)}| \leq d_\T(x,\parent(x))^p ,\quad x\in \X  \right\}.
      \label{eq:Phi*_trees_simple}
    \end{equation}

    \paragraph{Rewrite the target function} 
    We define linear operators $S_\T, D_\T : \RR^\X \ra \RR^\X$ by
    \[
      (D_\T v)_x = 
      \begin{cases}
        v_x - v_{\parent(x)} & x\neq\root(\T) \\
        v_{\root(\T)} & x=\root(\T).
      \end{cases}, \quad
      (S_\T u)_x = \sum_{x'\in\children(x)} u_{x'}.
    \]
    \begin{lem}
      \label{lem:S_D}
      For $\bu,\bv \in\RR^\X$ we have $\scp{\bu}{\bv} = \scp{S_\T \bu}{D_\T \bv}$.
    \end{lem}
    \begin{proof}
      We compute
      \begin{align*}
        \scp{S_\T \bu}{D_\T \bv} &= \sum_{x\in \X} (S_\T \bu)_x (D_\T \bv)_x \\
        & = \sum_{x\in \X\setminus \left\{ \root(\T) \right\}}
        \sum_{x'\in\children(x)} (v_x - v_{\parent(x)}) u_{x'} \\
        &\quad + \sum_{x'\in\children(\root(\T))} v_{\root(\T)}u_{x'} \\ 
        & = \sum_{x\in \X} \sum_{x'\in\children(x)} v_x u_{x'} \\
        &\quad - \sum_{x\in \X\setminus \left\{ \root(\T)
        \right\}}\sum_{x'\in\children(x)} v_{\parent(x)} u_{x'}\\ 
        & = \sum_{x\in \X} u_x v_x,
      \end{align*}
      which proves the Lemma. To see how the last line follows let
      $\children^1(x)$ be the set of immediate predecessors of $x$, that is
      children of $x$ that are connected to $x$ by an edge. Then we can write
      the second term in the second to last line above as
      \begin{align*}
        \sum_{x\in \X\setminus \{ \root(\T) \}}
        \sum_{x'\in\children(x)} v_{\parent(x)} u_{x'}
        & = \sum_{y\in\X} \sum_{x\in\children^1(y)}\sum_{x'\in\children(x)} v_y
        u_{x'} \\
        &= \sum_{y\in\X}\sum_{x'\in\children(y)\setminus \{y\}} v_y u_{x'}
      \end{align*}
      and the claim follows.
    \end{proof}
    If $\bu\in\Phi^*_p$, as given in \eqref{eq:Phi*_trees_simple}, we
    have for $x\neq\root(\T)$ that 
    \[
      |(D_\T \bu)_x| = |u_x - u_{\parent(x)}| \leq d_\T(x,\parent(x))^p.
    \]
    With these two observations and Lemma \ref{lem:S_D}, we get for
    $\bm{G}\sim\mathcal{N}(0,\Sigma(\br))$ and
    $\bu\in\Phi^*_p$ that 
    \begin{equation}
      \scp{\bm{G}}{\bu} = \scp{S_\T \bm{G}}{D_\T \bu} 
      \leq \sum_{\root(\T)\neq x \in \X} |(S_\T \bm{G})_x| d_\T(x,\parent(x))^p.
      \label{eq:tree_bound_on_max}
    \end{equation}
    Therefore, $\max_{\bu\in\Phi^*_p} \scp{\bm{G}}{\bu}$ is bounded by
    $\sum_{\root(\T)\neq x \in \X} |(S_\T \bm{G})_x|
    d_\T(x,\parent(x))^p$. Since $D_\T$ is an isomorphism, we can define a vector
    $\bv\in\RR^\X$ by 
    \[
      (D_\T\bv)_x = \sgn( (S_\T \bm{G})_x) d_\T(x,\parent(x))^p.
    \]
    From \eqref{eq:Phi*_trees_simple} we see that $\bv\in\Phi^*_p$ and Lemma
    \ref{lem:S_D}
    shows that $\scp{\bm{G}}{\bv}$ attains the upper bound in
    \eqref{eq:tree_bound_on_max}. 
    This concludes the proof.

    \subsection{Proof of Corollary \ref{cor:samworth}}
    In order to use Theorem \ref{thm:trees} we define the tree $\T$ with vertices $\left\{
    x_1,\dots, x_N \right\}$, edges $E=\left\{ (x_j, x_{j+1}),j =1,\dots,N-1
  \right\}$ and $\root(\T)=x_N$. Then, if $\bm{G}\sim\mathcal{N}(0,\Sigma(\br))$,
  we have that $\left\{ (S_\T\bm{G})_j \right\}_{j=1,\dots, N}$ is a Gaussian
  vector
  such that for $i\leq j$
  \begin{align*}
    & \cov( (S_\T\bm{G})_i, (S_\T\bm{G})_j)  = \sum_{\substack{k\leq i \\ l\leq
    j}} E\left[ G_k G_l \right]  
    = \sum_{k\leq i} r_k(1-r_k) - \sum_{\substack{k\leq i\\l\leq j\\
    k\neq l}} r_kr_l \\ 
    &= \bar{r}_i - \sum_{\substack{k\leq i \\ l\leq i}}r_kr_l -
    \sum_{\substack{k\leq i \\ i<l\leq j}} r_k r_l 
    = \bar{r}_i - \bar{r}_i^2 - \bar{r}_i(\bar{r}_j - \bar{r}_i)) 
    = \bar{r}_i - \bar{r}_i \bar{r}_j.
  \end{align*}
  Therefore, we have that for a standard Brownian bridge $B$
  \[
    S_\T\bm{G} \sim (B(\bar{r}_1), \dots , B(\bar{r}_N)).
  \]
  Together with $d(x_j, \parent(x_j))=(x_{j+1} - x_j)^2$, and
  \eqref{eq:weak_conv_trees} this proves the Corollary.